\documentclass[a4paper,11pt]{amsart}
\usepackage[english]{babel}
\usepackage{subcaption}
\usepackage{caption}
\usepackage[a4paper,margin=1in]{geometry}
\usepackage{enumerate}
\usepackage{amsmath}
\usepackage{amssymb}
\usepackage{amsthm}
\usepackage{graphicx}
\usepackage{dsfont}
\usepackage{comment}
\usepackage[colorlinks=true, allcolors=black]{hyperref}
\newtheorem{theorem}{Theorem}[section]
\newtheorem{lemma}{Lemma}[section]
\newtheorem{pro}{Proposition}[section]
\newtheorem{definition}{Definition}[section]
\newtheorem{remark}{Remark}[section]
\newtheorem{assume}{Assumption}[section]
\newtheorem{corollary}{Corollary}[section]
\usepackage{actuarialsymbol}
\numberwithin{equation}{section}
\def\ancomment#1{{\textcolor{red}{#1}}}
\def\shihaocomment#1{{\textcolor{blue}{#1}}}

\usepackage{booktabs}

\usepackage{amsaddr}
\usepackage{enumitem}
\setlist[itemize]{topsep=0pt, partopsep=0pt, parsep=0pt, itemsep=0pt}

\title{Striking the balance: Life Insurance Timing and Asset Allocation in Financial Planning} 

\author[1]{An Chen$^\ast$ \and Giorgio Ferrari$^\dagger$ \and Shihao Zhu$^\ddagger$}

\thanks{$^\ast$ Institute of Insurance Science, Ulm University, Helmholtzstr.\ 20, 89069 Ulm, Germany (an.chen@uni-ulm.de).}
\thanks{$^\dagger$ Center for Mathematical Economics (IMW), Bielefeld University, Universit{\"a}tsstrasse 25, 33615, Bielefeld, Germany (giorgio.ferrari@uni-bielefeld.de).}
\thanks{$^\ddagger$ Corresponding author. Institute of Insurance Science, Ulm University, Helmholtzstr.\ 20, 89069 Ulm, Germany (shihao.zhu@uni-ulm.de).}

\date{\today}

\begin{document}
  \maketitle



\begin{abstract}
This paper investigates the consumption and investment decisions of an individual facing uncertain lifespan within a Black-Scholes market framework. A key aspect of our study involves the agent's option to choose when to acquire life insurance for bequest purposes. We examine two scenarios: one with a \emph{fixed} bequest amount and another with a \emph{controlled} bequest amount. Applying duality theory and addressing free-boundary problems, we analytically solve both cases, and provide explicit expressions for value functions and optimal strategies in both cases. In the first scenario, where the bequest amount is fixed, distinct outcomes emerge based on different levels of risk aversion parameter $\gamma$: (i) the optimal time for life insurance purchase occurs when the agent's wealth surpasses a critical threshold if $\gamma \in (0,1)$, or (ii) life insurance should be acquired immediately if $\gamma >1$. In contrast, in the second scenario with a controlled bequest amount, regardless of $\gamma$ values, immediate life insurance purchase proves to be optimal. {Finally, we extend the analysis to consider a scenario in which the individual earmarks part of her initial wealth for inheritance, where a critical wealth threshold consistently emerges.}
\end{abstract}


\vspace{2mm}


\vspace{2mm}

\vspace{2mm}

\vspace{2mm}

\section{Introduction}

In the literature on optimal consumption and investment decisions, as discussed in \cite{richard1975optimal} and the literature review provided in the penultimate paragraph of this introduction, there has been thorough exploration of integrating life insurance. The emphasis is on investigating the demand for life insurance and comprehending its influence on consumption and portfolio choices across an individual's life cycle. Nevertheless, in alignment with the groundbreaking work of \cite{richard1975optimal}, the existing theoretical literature predominantly focuses on determining the optimal sum insured for the policy rather than investigating the optimal \emph{timing} for purchasing  insurance.  

Determining the optimal time to purchase life insurance is of utmost importance.
	The timing of purchasing life insurance is relevant because it affects the cost of premiums and the insurability. When considering an individual's entire lifespan, it may not be feasible for them to purchase life insurance or build an estate for their heirs at a young age due to limited capital. Even if an individual has sufficient capital, purchasing life insurance too early may adversely affect their standard of living during their youth. On the other hand, purchasing life insurance at a young age often results in lower premiums, reducing the total amount the individual will need to spend over their lifetime. As premiums tend to be higher for older policyholders, it may not be optimal to delay acquiring life insurance until later in life. This is because the risk of mortality increases with age. To sum up, by purchasing life insurance at a younger age, you can lock in lower premium rates for the duration of the policy. Starting early can save you money in the long run and make coverage more affordable. In addition, as we human beings progress through life, our financial responsibilities tend to increase. This may include getting married, starting a family, purchasing a home, or taking on significant debts. Analyzing when to buy life insurance allows us to align the coverage with our changing financial responsibilities.

In this paper, we explore the optimal timing for purchasing life insurance in the context of an agent facing uncertain lifetime and stochastic labor income. The agent can continuously invest in a Black-Scholes market and decides when to buy life insurance for bequest purposes. Two cases for modeling bequest are considered: one where bequest is a predetermined amount, and the other where the policyholder can choose the bequest amount as an additional choice variable. 

{The \emph{predetermined} bequest amount remains highly relevant in practice, despite the prevalence of the controlled bequest model in the literature. On the one hand, fixed bequest allocations foster a perception of fairness---particularly when an estate is divided equally among multiple children---a strategy that over two-thirds of testate decedents with multi-child families in US employ \cite{bernheim2003bequests}. Moreover, fairness can be tailored to individual circumstances: for instance, a testator might grant extra support to a child with special needs or financial difficulties (e.g.,\ \cite{elinder2018inheritance, hamaaki2019intra}.) On the other hand, specifying exact bequest amounts gives testators clear control over their estate's distribution, ensuring that their wishes are both unambiguous and honored. As \cite{kopczuk2013taxation} note, explicit allocations often align with broader goals---such as preserving family assets or guiding beneficiary behavior over time. In addition, predetermined bequests simplify the administrative process and improve estate-planning efficiency, especially when one accounts for legal and tax considerations (see e.g.,\cite{poterba2001estate, mcgarry2013estate}.)}

{In contrast, the \emph{controlled bequest model} offers policyholders the flexibility to adjust the bequest amount according to changing personal circumstances and strategic considerations. By incorporating the bequest as an additional choice variable, policyholders can dynamically respond to evolving financial conditions, family needs, or market developments.  Specifically,  most existing models in the literature (see, e.g., \cite{pliska2007optimal}) explicitly frame bequests as functions dependent on the life insurance premiums and the force of mortality (as we shall do in Section \ref{sec:conbeq}).  
This flexibility allows for a more personalized approach to estate planning. For instance, rather than committing to a fixed amount, a policyholder might decide to allocate a larger portion of the estate during periods of financial stability or, alternatively, reserve resources when future uncertainties loom.}

From a mathematical point of view, we model the previous problem as a \textit{random time horizon, two-dimensional stochastic control problem with discretionary stopping}. The two coordinates of the state process are the wealth process $X$ and the stochastic labor income $Y$. 
Moreover, the utility increases after purchasing life insurance due to the bequest motives. The agent's aim is to choose consumption rate $c$, portfolio $\pi$, and the life insurance purchase time $\eta$ in order to maximize the total expected utility, up to the random death time $\tau$.\footnote{Problems with a similar structure arise, for instance, in retirement time choice models, where the agent has to consume and invest in risky assets, and to decide when to retire (see, e.g.,\ \cite{jin2006disutility}). Combined stochastic control/optimal stopping problems also arise in Mathematical Finance, namely, in the context of pricing American contingent claims under constraints and utility maximization problem with discretionary stopping; see, e.g., \cite{karatzas1998hedging} and \cite{karatzas2000utility}.} To address the intricate mathematical structure of our problem, where the interplay of consumption, portfolio choice, and life insurance purchase is nontrivial, we employ a combination of duality and a free-boundary approach. \\

\noindent Notably, we find distinct results for the two cases of bequest: 
\begin{itemize}
		\item When the bequest amount is predetermined, the optimal timing for life insurance purchase depends significantly on the relative risk aversion. For a relative risk aversion greater than 1, it is optimal to acquire life insurance immediately. For a relative risk aversion between 0 and 1, the optimal time to buy life insurance is not a corner solution, but rather corresponds to the point where the policyholder's wealth reaches an early exercise boundary at each moment in time. If the wealth level exceeds the early exercise boundary, purchasing life insurance becomes advantageous. Conversely, if the wealth remains below the early exercise boundary, it is more beneficial not to buy life insurance. Additionally, we can determine the associated optimal consumption and investment strategies.
		\item On the contrary, when the policyholder has freedom to independently choose the bequest amount, irrespective of relative risk aversion, the optimal decision is always to purchase insurance without delay. In cases where relative risk aversion falls within the range of (0, 1), the optimal timing for acquiring life insurance differs from situations where the bequest amount is predetermined. With a fixed bequest amount, a trade-off emerges between the premium cost and the utility derived from the bequest, leading to the establishment of a boundary dictating the optimal purchase time. However, when the agent has the flexibility to determine the bequest amount, this trade-off can be mitigated by strategically selecting the bequest amount based on their wealth at the time of purchase; in other words, the agent will always choose the optimal bequest amount (the corresponding premium cost) at the purchasing time. Therefore, early acquiring the life insurance would obtain more utility from the bequest amount. Consequently, the optimal strategy becomes an immediate purchase of life insurance, maximizing utility derived from the bequest amount. 
	\end{itemize}
	

{Besides the main results derived from our base models (cf.\ Sections \ref{sec:predeq} and \ref{sec:conbeq}), we also discuss several extensions in Section \ref{dissec}. Specifically, instead of our initial setup—where the individual determines all consumption and life insurance decisions at time $t=0$ with an initial wealth of $x$—we consider an alternative scenario in which the policyholder sets aside a portion of their initial wealth for inheritance.  
Additionally, we extend our model by incorporating an age-dependent force of mortality, moving beyond the constant force of mortality assumption in our base case. This enhancement makes our model more realistic. Interestingly, in both extensions, we consistently identify a critical wealth threshold (which becomes age-dependent in the case of an age-dependent force of mortality). Above this threshold, purchasing life insurance remains a rational choice, regardless of the values of $\gamma$.} 
	

As previously indicated, there exists a considerable body of literature integrating life insurance decisions into optimal consumption and asset allocation frameworks. Our research contributes to this existing literature by examining the optimal timing for acquiring life insurance within the context of a life-cycle consumption and portfolio planning problem.
\cite{richard1975optimal} introduces the notion of optimal consumption and asset allocation when confronted with uncertain life expectancy.
Building upon Merton's optimal consumption and investment problem in a continuous-time framework (see \cite{merton1971optimum}), \cite{richard1975optimal} expands the scope by incorporating the investor's arbitrary yet known lifetime distribution, addressing a broader range of scenarios. Surprisingly, the study unveils that the investment strategy remains consistent even when compared to situations with specific lifetime assumptions.
The framework proposed by \cite{richard1975optimal} has since been extended in various directions. Both empirical and theoretical studies exploring the demand for life insurance have grown, with empirical investigations focusing on discerning factors influencing life insurance consumption (e.g., \cite{li2007demand}, \cite{braun2016consumer} and references therein). On the theoretical front, \cite{pliska2007optimal} look into an optimal life insurance and consumption problem for an income earner when the lifetime random variable is unbounded. Subsequently, \cite{huang2008portfolio} address a portfolio choice problem incorporating mortality-contingent claims and labor income under general HARA utilities. \cite{bayraktar2013life} consider a case of exponential utility, determining the optimal amount of life insurance for a household with two wage earners. Expanding on this, \cite{wei2020optimal} investigate a scenario where a couple with correlated lifetimes seeks to optimize their consumption, portfolio, and life insurance purchasing strategies to maximize the family objective until retirement. Recently, \cite{wang2021household} study optimal decisions on purchasing life insurance for a household with two consecutive generations, incorporating uncertain income and market ambiguity. 
 Furthermore, \cite{park2023robust} examine a robust consumption-investment problem involving retirement and life insurance decisions for an agent concerned about inflation risk and model ambiguity.
 

The reminder of the paper is structured as follows:
Section~\ref{sec:mod} introduces the underlying financial market and formulates the optimization problem encompassing consumption, investment, and the timing of life insurance purchase. Section~\ref{sec:predeq} solves the optimization problem when the bequest amount is predetermined and provides the corresponding numerical illustrations.  Section~\ref{sec:conbeq} addresses the optimization problem when the bequest amount is an additional choice variable and gives the numerical examples. Section \ref{dissec} provides some further discussions and extensions to the base model. Section~\ref{sec:con} concludes the paper. { Appendix \ref{appendix}     offers additional technical details related to the results discussed in  Sections \ref{sec:predeq}, \ref{sec:conbeq}} and \ref{dissec}.

\section{The model} \label{sec:mod}
Let $({\Omega}, {\mathcal{F}}, \mathbb{F}:=\{ \mathcal{F}_t, t \geq 0  \}, {\mathbb{P}})$ be a complete filtered probability space, on which it is defined a strictly positive random variable $\tau$ independent of $\mathcal{F}_\infty$. Think of an $x_0$-year economic agent who considers to buy a life insurance contract for bequest motive. We assume that the random remaining lifetime  $\tau$ of the agent is distributed exponentially with a constant force of mortality\footnote{{In Section \ref{dissec}, we consider an age-dependent force of mortality as an extension.}}  $m>0$; that is, for any $t\geq 0$
\begin{align*}
{}_tp_{x_0}= \mathbb{P}(\tau \geq t) =e^{-mt}
\end{align*}
is the probability that the agent with age $x_0$ survives at least $t$ years. 

 If the agent purchases life insurance at time $t$, we assume that she would like to leave a target amount $B_t\geq 0$ to her heirs upon death, and would like to choose the optimal time $\eta$ to invest in the life insurance. By investing in the life insurance, we assume that she needs to pay a continuous premium $h_t$ until death. Then at the purchase time $t$, the premium $h_t$ is determined such that  
\begin{align}\label{premium0}
    \int^\infty_t h_t e^{-r(s-t)} \px[s-t]{x_0+t} ds =-\int^\infty_t B_t  e^{-r(s-t)} \frac{\partial}{\partial s}\px[s-t]{x_0+t} ds,
\end{align}
where $\px[s-t]{x_0+t}$, $s \geq t$, can be interpreted as the conditional survival probability that an $x_0$-year today survives $x_0+s$ given that she has survived $x_0+t$. 
 In terms of the constant force of mortality, from (\ref{premium0}) we have 
\begin{align}\label{premium}
    h_t =B_t \cdot m. 
    \end{align}
In this paper we will consider two different cases: a predetermined bequest amount and a controlled bequest amount. In the former case, the bequest amount $B_t:=B$ is a predetermined constant, whereas $\{B_t, t\geq 0\}$ is an $\mathcal{F}_t$-adapted control variable in the latter case.  

We assume that the agent also invests in a financial market with two assets. One of them is a risk-free bond, whose price $S^0:=\{ S^0_t, t \geq 0  \}$ evolves as 
\begin{align*}
dS^0_t = rS^0_tdt, \quad S^0_0 =s^0>0,
\end{align*}
where $r>0$ is a constant risk-free rate. 
The second one is a stock, whose price is denoted by $S:=\{S_t, t \geq 0\}$ and it satisfies the stochastic differential equation
\begin{align*}
dS_t &= \mu S_t dt + \sigma S_t dW_t, \quad S_0 =s>0,
\end{align*}
where $\mu>r$ and $\sigma>0$ are given constants. Here, $W:=\{ W_t, t \geq 0 \}$ is an $\mathbb{F}$-adapted standard Brownian motion on $(\Omega,\mathcal{F},\mathbb{P})$.

As highlighted by \cite{campbell1980demand}, stochastic income is a significant factor in household economic choices, including the decision to purchase life insurance. In this study, we assume that the stochastic income is contingent on the individual's mortality risk and the uncertainty inherent in the financial market. More specifically, we assume that the agent receives a stochastic labor income $Y:=\{Y_t,t\geq 0\}$ as long as she is alive. We assume the labor income to be spanned by the market and is driven by the same Brownian motion as the stock (see also \cite{dybvig2010lifetime}): 
\begin{align}\label{income}
Y_t=y e^{(\mu_y -\frac{1}{2}\sigma_y^2)t+ \sigma_y  W_t}, \ \text{for} \ t > 0, \ Y_0=y>0.
\end{align}
Here $\mu_y \in \mathbb{R}$ and $\sigma_y>0$ are constants, representing the instantaneous growth rate and the volatility of the labor income, respectively.\footnote{Individuals who do not receive variable salaries tied directly to the financial market risk may not have an income that can be fully hedged through tradable assets. Therefore, this assumption that labor income risk is perfectly correlated with the stock market does not present a fully realistic one. However, we choose to integrate income into our analysis, notwithstanding the somewhat less realistic assumption regarding income risk, to acknowledge the significance of human capital in life-cycle decisions. Simultaneously, we must acknowledge that this is done to establish a framework that is mathematically tractable, similar to those presented in \cite{dybvig2010lifetime} for modeling income risk. We leave the case with not fully hedgeable income risk (e.g. by adding another Brownian motion of labor income that is not perfectly correlated) for future research. }          

We define the market price of risk $\theta:= \frac{\mu-r}{\sigma}$ and the state-price-density process $\xi_{t}:=e^{-rt-\theta W_t -\frac{1}{2}\theta^2 t}$. Since the labor income is perfectly correlated with the market, the present value of  future labour income at time $t$, $g_t$, under the assumption that the agent is always alive, is such that
\begin{align}\label{gt}
g_t&:= \mathbb{E} \bigg[\int^{\infty}_t \frac{\xi_s}{\xi_t} Y_s ds \bigg|\mathcal{F}_t \bigg]=Y_t \cdot \mathbb{E} \bigg[\int^{\infty}_t \frac{\xi_s}{\xi_t}  \frac{Y_s}{Y_t} ds\bigg|\mathcal{F}_t \bigg]=Y_t \cdot \bigg[\int^{\infty}_t e^{-(r-\mu_y+\sigma_y \theta)(s-t)}  ds\bigg] \nonumber\\
&=\left \{
\begin{aligned}
& \frac{Y_t}{\kappa}, \ &\text{if} \ \kappa > 0, \\
&\infty, \ &\text{if} \ \kappa \leq  0,
\end{aligned}
\right. 
\end{align} 
with $\kappa:=r-\mu_y+\sigma_y\theta$ being the effective discount rate for labor income. Throughout the paper we assume $\kappa>0$.


The agent also consumes from her wealth, while investing in the financial market. Denoting by $\pi_t$ the amount of wealth invested in the stock at time $t$, the agent then chooses $\pi_t$ as well as the consumption process $c_t$ at time $t$. Therefore, the agent's wealth $X:= \{ X^{c,\pi, \eta}_t, t\geq 0 \}$ evolves as 
\begin{align}\label{wealth}
d X^{c,\pi, \eta}_t =  [\pi_t(\mu-r)+rX^{c,\pi,\eta}_t-c_t-h_t \mathds{1}_{\{t \geq \eta \}}+Y_t]dt +\pi_t \sigma dW_t, \quad X_0^{c,\pi, \eta}= x,
\end{align}
where $h_t$ is given by (\ref{premium}), and the $\mathbb{F}$-stopping time $\eta$ is the life insurance purchase, after which a continuous stream of premium payments will be paid. In the following, we shall simply write $X$ to denote $X^{c,\pi, \eta}$, where needed.  

The agent determines the optimal levels of consumption and investment, as well as the timing for purchasing life insurance to secure the bequest. We apply a power utility function for both consumption and bequest, specifically, 
\begin{align}\label{utility}
u(\cdot):=\frac{(\cdot)^{1-\gamma}}{1-\gamma}, \ \gamma >0 \ \text{and} \ \gamma \neq 1.
\end{align}
The agent's aim is then to maximize the expected lifetime utility\footnote{{In Section \ref{dissec}, we consider another related scenario where the policyholder allocates a portion of their initial wealth for inheritance as an extension. Consequently, $u(0)$ in (\ref{objective}) becomes $u(q)$, where $q$ is directly designated for inheritance.}}
\begin{align}\label{objective}
    \mathbb{E}\bigg[ \Big( \int^\tau_0e^{-\rho t}u(c_t)dt +e^{-\rho \tau}u(0) \Big)\mathds{1}_{\{ \eta \geq \tau\}} + \mathds{1}_{\{\eta<\tau\}}\Big( \int^\tau_0 e^{-\rho t}u(c_t)dt+e^{-\rho \tau}u(lB_\tau) \Big) \bigg],
\end{align}
where $\rho>0$ is a constant representing the subjective discount rate. For $\eta \geq \tau$, no life insurance is taken out during the individual's lifetime. For an individual with a bequest motive, it can bring a disutility of zero for $\gamma \in (0,1)$, i.e. $u(0)=0$, or an infinite disutility level for $\gamma >1$, i.e. $u(0)=-\infty$. The constant $l >0$ measures how the agent weighs bequest in her total lifetime utility. 

{\raggedleft{Thanks to}} Fubini's Theorem and independence between $\tau$ and $\mathcal{F}_\infty$, we can disentangle the market and mortality risk and write (\ref{objective}) as
\begin{align}\label{objective2}
    &\mathbb{E}\bigg[\int^\infty_0e^{-(\rho+m) t}u(c_t)dt +u(0) \int^\eta_0 e^{-(\rho+m) t} \, m \,dt+   \int^\infty_\eta e^{-(\rho+m) t}\, m \, u(lB_t) dt  \bigg]\nonumber\\
    &=  \mathbb{E}\bigg[\int^\eta_0e^{-(\rho+m) t}\Big(u(c_t)+m \, u(0) \Big)dt +\int^\infty_\eta e^{-(\rho+m) t} \Big(u(c_t)+m \,u(lB_t)\Big) dt  \bigg].
\end{align}

\section{A predetermined bequest amount} \label{sec:predeq} 
\subsection{Problem formulation}
In this section we consider the model with a predetermined bequest amount $B_t=B$. In this case, the premium in (\ref{premium}) is $h_t =m \cdot B:=h$. 

Here and in the sequel, for $y>0$, we write $\mathcal{O}:= (-\frac{y}{\kappa},\infty) \times \mathbb{R}_+$  with $\mathbb{R}_+:=(0,\infty)$.  We denote by $\mathcal{S}$ the class of $\mathbb{F}$-stopping times $\eta: \Omega \to [0,\infty]$. Then we introduce the set of admissible strategies $\mathcal{A}(x,y)$ as it follows.

\begin{definition}\label{admissiblecontrol}
Let $(x,y) \in \mathcal{O}$ be given and fixed. The triplet of choices $(c,\pi, \eta)$ is called an \textbf{admissible strategy} for $(x,y)$, and we write $(c,\pi, \eta) \in \mathcal{A}(x,y)$, if it satisfies the following conditions: 
\begin{enumerate}[label=(\roman*)]
	\item  $c$ and $\pi$ are progressively measurable with respect to $\mathbb{F}$, $\eta \in \mathcal{S}$;
	\item $c_t \geq 0$ for all $t \geq 0$ and $\int_0^t(c_s+|\pi_s|^2)ds <\infty$ for all $t\geq 0$ $ \ \mathbb{P}$-a.s.; 
	\item $X^{c,\pi, \eta}_t +g_t >  \frac{h}{r}\mathds{1}_{\{t\geq  \eta\}}$ for all $t\geq 0$, where $g_t$ is defined in (\ref{gt}).
\end{enumerate}
\end{definition}
{\raggedleft{The term}} $\frac{h}{r}=\int^\infty_t h e^{-r(s-t)}ds$ in Condition (iii) is the present value of the future premium payment of the agent, under the assumption that the agent is always alive. Due to (iii) above, the agent is able to consume and invest as long as her wealth level plus her present value of the future income is above $\frac{h}{r}$ at time $t\geq \eta$. It also means that we allow the agent to borrow fully against the stream of future income. Before purchasing life insurance, she should keep her wealth plus present value of the future income positive for further consumption or financial investment. {This assumption of borrowing against the future income is equivalent to say that the economic agent possesses her human capital, defined as the present value of future income, which she can use for consumption and investment choices. Thus, an individual's overall wealth encompasses both their current financial assets and their human capital. As extensively pointed out in academic literature on life-cycle decisions, e.g. \cite{cocco2005consumption} and \cite{dybvig2010lifetime}, human capital, representing one's capacity to work, serves as the primary asset. In youth, the agent can capitalize on their future income to allocate more towards risky assets, leading to a higher equity allocation (see e.g. \cite{bosserhoff2022investment})}.

From (\ref{objective2}), given the Markovian setting, the agent aims at determining
\begin{align}\label{vxy}
V(x,y):=&\sup_{(c, \pi,\eta) \in \mathcal{A}(x,y) }\mathbb{E}_{x,y}\bigg[\int^\eta_0e^{-(\rho+m) t}\Big(u(c_t)+m \,u(0)\Big)dt \nonumber\\  &+ \int^\infty_\eta e^{-(\rho+m) t} \Big(u(c_t)+m\,u(lB)\Big) dt  \bigg],\nonumber\\
=:&\sup_{(c,\pi,\eta) \in \mathcal{A}(x,y)}J_{x,y}(c,\pi,\eta).
\end{align}
Here, $ \mathbb{E}_{x,y}$ represents the expectation under $\mathbb{P}_{x,y}$, specifically conditioned on $X_0=x$ and $Y_0=y$. Throughout the remainder of this paper, given $\boldsymbol{\psi} \in \mathbb{R}^n, n\geq 1$,  $\mathbb{E}_{\boldsymbol{\psi}}$ will denote the expectation under the measure $\mathbb{P}_{\boldsymbol{\psi}}$, where $\mathbb{P}_{\boldsymbol{\psi}}$ is the probability measure on $(\Omega,\mathcal{F})$ under which the considered Markov process $\{\Psi_t, t\geq 0 \}$ starts at time zero from the specified level $\boldsymbol{\psi}$. In the sequel, whenever necessary, we also write $X^x$ (similarly, $Z^z, Y^y, X^{*x}$) to stress the dependency of the considered processes on their initial datum. 

The rest of this section will study $(\ref{vxy})$. To that end, we make the following assumption.
\begin{assume}\label{assume2}
We assume $\rho+m > (1-\gamma)r+ \frac{1-\gamma}{2\gamma} \theta^2$.
\end{assume}

The above assumption is a standard assumption to make the optimization problem well-defined and holds throughout the paper without further comments. Specifically, under Assumption \ref{assume2} when the agent is forced to choose $\eta=0$, $J_{x,y}(c,\pi,0)$ is finite (see Propositions \ref{pro3.1} and \ref{pro4.1}  below for any choice of admissible $(c,\pi)$). Note that the above assumption holds always for a relative risk aversion level larger than 1 and $r \geq 0$. 
  Similarly, it can be shown that when $\eta=\infty$ (that is, the agent never buys the life insurance), $J_{x,y}(c,\pi,\infty)$, is also finite for any admissible $(c,\pi)$ thanks to Assumption \ref{assume2}. A similar requirement is posed in e.g., \cite{karatzas1986explicit} and \cite{jin2006disutility}.

The following theorem holds directly by observing (\ref{vxy}), since $u(0)=-\infty$ when $\gamma>1$.
\begin{theorem}\label{buy}
When $\gamma>1$, the optimal purchasing time is $\eta^*=0$.
\end{theorem}
In other words, when $\gamma>1$, the optimal decision for the agent is to purchase life insurance immediately. Therefore, we just need to consider the case $\gamma<1$ in the following subsections.

\subsection{Solution to the problem for $\gamma \in(0,1)$}\label{sec3-2}
In this subsection, we determine the explicit solution to (\ref{vxy}) when $0<\gamma<1$ by combining a duality and a free-boundary approach. To accomplish that, we shall first conduct successive transformations (cf.\ Subsections \ref{3-2-1}, \ref{3-2-2} and \ref{3-2-3}) that connect the original stochastic control-stopping problem (with value function $V$) into its dual problem (with value function $v$) by martingale and duality methods (similar to \cite{karatzas2000utility}). Then we study the reduced-version dual stopping problem (with value function $\widehat{v}$) and obtain the explicit forms of the free boundary and of the value function by using the classical ``guess and verify" approach  in Subsections \ref{3-2-4} and \ref{3-2-5}. If the reader is not interested in the detailed mathematical analysis, she can skip this section and find the optimal policies in Theorem \ref{optimal} and the numerical illustrations in Subsection \ref{sec3num}, prior to reaching Section \ref{sec:conbeq}.  

\subsubsection{The static budget constraint}\label{3-2-1}
We now transform the dynamic budget constraint in (\ref{wealth}) with $h=h_t=m \cdot B$ into a static budget constraint by using the well-known method developed by \cite{karatzas1987optimal} and \cite{cox1989optimal}.

From the optional sampling theorem and Fatou's lemma, we can express the dynamics of the agent's wealth \eqref{wealth} through the following static budget constraint. This constraint applies to two scenarios: one before and the other after the acquisition of life insurance:
\begin{align}\label{static1}
		\mathbb{E}_{x,y}\big[{\xi_{t \wedge \eta}}X_{t \wedge \eta}\big]+\mathbb{E}_{x,y}\bigg[ \int_0^{t \wedge \eta} \xi_{u} c_u du                         \bigg]  \leq x+ \mathbb{E}_{x,y}\bigg[ \int_0^{t \wedge \eta} \xi_{u} Y_u du \Bigg], \quad  \ t \geq 0,
	\end{align}
and
\begin{align}\label{static2}
\mathbb{E}_{x,y}\big[{\xi_{t \vee \eta}}X_{t \vee \eta }\big]+\mathbb{E}_{x,y}\bigg[ \int_\eta^{t \vee \eta} \xi_{u} \Big(c_u+h\Big) du                         \bigg]   \leq X_\eta+ \mathbb{E}_{x,y}\bigg[ \int_\eta^{t \vee \eta} \xi_{u} Y_u du \Bigg], \quad \ t \geq 0.
\end{align}

\subsubsection{The optimization problem after purchasing life insurance}\label{3-2-2}
In this subsection we will consider the agent’s optimization problem \emph{after} purchasing life insurance, and over this time period only consumption and portfolio choice have to be determined. Formally, the model in the previous section accommodates to this case if we let $\eta =0$, where $0$ is the fixed starting time. Then, letting $\mathcal{A}_0(x,y):= \{ (c,\pi): (c,\pi,0) \in \mathcal{A}(x,y)\},$ where the subscript $0$ indicates that the purchasing time $\eta$ is equal to $0$, the agent's value function after purchasing life insurance reads (simply let $\eta=0$ in (\ref{vxy}))   
\begin{align}\label{vhat}
    \widehat{V}(x,y):= \sup_{(c,\pi)\in \mathcal{A}_0(x,y)} \mathbb{E}_{x,y}\bigg[\int^\infty_0 e^{-(\rho+m) s } \Big(u(c_s)+m \, u(lB)\Big) ds \bigg].
    \end{align}
Recalling that $\xi_{s}=e^{-rs-\theta W_s -\frac{1}{2}\theta^2 s}$, and for any pair $(c,\pi)\in \mathcal{A}_0(x,y)$ with a Lagrange multiplier $z>0$, we have
\begin{align}\label{dual}
 \mathbb{E}&_{x,y} \bigg[    \int_0^{\infty}  e^{-(\rho+m) s } \Big(u(c_s)+m  \, u(lB)\Big)  ds  \bigg] \nonumber \\ \leq &\
 \mathbb{E}_{x,y} \bigg[    \int_0^{\infty} e^{-(\rho+m) s } \Big(u(c_s)+m  \, u(lB)\Big)  ds  \bigg]   - z      \mathbb{E}_{x,y}\bigg[  \int_0^{\infty} \xi_{s} (c_s-Y_s) ds                         \bigg]  +z( x-\frac{h}{r})                  \nonumber               \\
 =&\   \mathbb{E}_{x,y}\bigg[    \int_0^{\infty }  e^{-(\rho+m) s }\Big(u(c_s)+m \, u(lB)\Big)  ds \bigg]\nonumber\\&- \mathbb{E}_{x,y}\bigg[ \int_0^{\infty}  e^{-(\rho+m) s }z P_s  (c_s-Y_s) ds             \bigg]+z( x-\frac{h}{r})  \nonumber  \\
\leq &\  \mathbb{E}_{x,y}\bigg[    \int_0^{\infty } e^{-(\rho+m) s }\Big(\widehat{u}(zP_s)+ zP_sY_s+m \, u(lB)\Big) ds  \bigg]+ z( x-\frac{h}{r}), 
\end{align}
where the first inequality results from the budget constraint stated in \eqref{static2} with $\eta=0$. Further, 
\begin{align}\label{pt}
P_t:= \xi_{t}e^{(\rho+ m)t } \quad \text{and} \quad \widehat{u}(z):= \sup_{c \geq 0}[u(c)-c z], \ z>0,
\end{align}
{\raggedleft{where}} $\widehat{u}(z)$ is the convex dual of $u(c)$. Let then $Z_t :=z P_t.$ By It\^o's formula, we obtain that the dual variable $Z$ satisfies
\begin{align}\label{Z}
dZ_t=  (\rho-r+ m)Z_tdt - \theta Z_t dW_t, \quad Z_0= z,
\end{align}
and we set 
\begin{align}\label{W}
\widehat{Q}(z,y):=  \mathbb{E}_{z,y}\bigg[    \int_0^{\infty }  e^{- (\rho+ m) s} \Big(\widehat{u}(Z_s)+Z_sY_s+m \, u(lB)\Big) ds  \bigg].
\end{align}

\begin{pro}\label{pro3.1}
$\widehat{Q}$ is finite and one has $\widehat{Q}\in  C^{2,2}(\mathbb{R}^2_+).$ Moreover, $\widehat{Q}$ satisfies 
\begin{align}\label{LW}
-{\mathcal{L}}\widehat{Q}=\widehat{u}+m\,u(lB)+zy, \ \text{on} \ \mathbb{R}^2_+,  
\end{align}
where 
\begin{align}\label{L}
{\mathcal{L}}\widehat{Q}:=\frac{1}{2}\theta^2 z^2\widehat{Q}_{zz}+(\rho-r+m)z\widehat{Q}_z +\frac{1}{2}\sigma_y^2 y^2\widehat{Q}_{yy}+\mu_y y\widehat{Q}_y-\theta \sigma_y z y\widehat{Q}_{zy}   -(\rho+m)\widehat{Q}.
\end{align}
\end{pro}
\begin{proof}

The proof is given in Appendix \ref{proofpro3.1}.

\end{proof}

{\raggedleft{It}} is then possible to relate the agent's value function $\widehat{V}$ after purchasing the life insurance to $\widehat{Q}$ through the following duality relation.
\begin{theorem}\label{dualrelation1}
For given $(x,y) \in (\frac{h}{r}-\frac{y}{\kappa},\infty)\times \mathbb{R}_+$, the following dual relations hold:
\begin{align*}
\widehat{V}(x,y) = \inf_{z >0} [ \widehat{Q}(z,y)+z(x-\frac{h}{r}) ], \quad \widehat{Q}(z,y) =  \sup_{x> \frac{h}{r}-\frac{y}{\kappa}}[\widehat{V}(x,y)-z(x-\frac{h}{r})].
\end{align*}
Moreover, optimal consumption process $c_t^*=(Z_t^*)^{-\frac{1}{\gamma}}$, optimal portfolio process $\pi^*_t= \frac{\theta (Z_t^*)^{-\frac{1}{\gamma}}}{K \gamma \sigma}-\frac{\sigma_y Y_t}{\sigma \kappa}$ and corresponding optimal wealth process $X_t^*=\frac{(Z^*_t)^{-\frac{1}{\gamma}}}{K}+\frac{h}{r}-\frac{Y_t}{\kappa}$, where $Z_t^*$ is the solution to Equation (\ref{Z}) with the initial condition $z^*$ satisfying $\widehat{Q}_z(z^*,y)+x-\frac{h}{r}=0$.

\end{theorem}
\begin{proof}
The proof is given in Appendix \ref{proofdualrelation1}.
\end{proof}

\subsubsection{The dual optimal stopping problem}\label{3-2-3}
Recall that we are focusing on the case $\gamma<1$ (i.e.\ $u(0)=0$) due to Theorem \ref{buy}. Notice that, for any $(x,y) \in \mathcal{O}$, it holds by definition of $V$ and strong Markov property that 
\begin{align}\label{vxy2}
    V(x,y)&\leq  \sup_{(c,\pi, \eta) \in \mathcal{A}(x,y)} \mathbb{E}_{x,y}\bigg[\int^\eta_0e^{- (\rho+ m) t}u(c_t) dt +e^{- (\rho+ m) \eta}\widehat{V}(X_\eta,Y_\eta)  \bigg].
    \end{align}
    

Now, for any $(x,y)\in \mathcal{O}$ and Lagrange multiplier $z>0$, from (\ref{vxy}), the budget constraint (\ref{static1}), (\ref{vxy2}), and recalling $P_t$ as in (\ref{pt}), we have 
\begin{align}\label{dual2}
 \mathbb{E}&_{x,y}\bigg[\int^\eta_0 e^{-(\rho+m)t }u(c_t)dt +\int^\infty_\eta e^{-(\rho+m)t } \Big(u(c_t)+m\,u(lB)\Big) dt  \bigg]
 \nonumber \\ \leq &\
\sup_{(c,\pi,\eta) \in \mathcal{A}(x,y)} \mathbb{E}_{x,y}  \Big[\int^\eta_0 e^{-(\rho+m) t }u(c_t)dt + e^{-(\rho+m) \eta }\widehat{V}(X_\eta,Y_\eta)  \Big] \nonumber \\& - z      \mathbb{E}_{x,y} \bigg[   \xi_{\eta}X_{\eta}+ \int_0^{\eta} \xi_{t}(c_t-Y_t) dt                         \bigg]  +z x                        \nonumber       \\
=&\   \sup_{(c,\pi,\eta) \in \mathcal{A}(x,y)} \mathbb{E}_{x,y}\bigg[    \int_0^{\eta}   e^{-(\rho+m) t } \bigg(u(c_t) -zP_tc_t+zP_tY_t\bigg) dt  \nonumber \\&+ e^{-(\rho+m) \eta } \widehat{V}(X_{\eta},Y_{\eta}) -  e^{-(\rho+m) \eta }zP_{\eta}X_{\eta}      \bigg]+z x  \nonumber \\
 \leq &\  \sup_{\eta \in \mathcal{S}} \mathbb{E}_{z,y}\bigg[    \int_0^{\eta} e^{-(\rho+m) t }\Big(\widehat{u}(Z_t)+Z_tY_t\Big) dt
+  e^{-(\rho+m) \eta }\Big(\widehat{Q}(Z_{\eta},Y_{\eta})-Z_{\eta}\frac{h}{r} \Big) \bigg]+ zx , 
\end{align}
where we recall that $\widehat{u}(z)= \sup_{c \geq 0}[u(c)-c z], z>0$, and $Z_t$ is defined in (\ref{Z}).

{\raggedleft{Hence, defining the (candidate) dual value function}}
\begin{align}\label{jzy}
v(z,y)&:=\sup_{ \eta \in \mathcal{S}
} \mathbb{E}_{z,y}\bigg[    \int_0^{\eta }  e^{- (\rho+ m) t}\Big(\widehat{u}(Z_t) +Z_tY_t\Big)dt + e^{- (\rho+ m) \eta}\Big(\widehat{Q}(Z_{\eta},Y_\eta)-Z_{\eta}\frac{h}{r} \Big) \bigg],
\end{align}
we have a two-dimensional optimal stopping problem, with dynamic $(Z,Y)$ as in (\ref{Z}) and (\ref{income}).


\subsubsection{Preliminary properties of the value function}\label{3-2-4}
To study the optimal stopping problem (\ref{jzy}), it is convenient to introduce the function
\begin{align}\label{4-3}
\widehat{v}(z,y):={v}(z,y)-(\widehat{Q}(z,y)-\frac{h}{r}z).
\end{align}

Applying It\^o's formula to $\{ e^{-(\rho+m)t }[\widehat{Q}(Z_t,Y_t)-Z_t \frac{h}{r}] , t \in [0,\eta]  \}$, and taking conditional expectations we have 
\begin{align*}
&\mathbb{E}_{z,y}\Big[ e^{-(\rho+m) \eta }\Big(\widehat{Q}(Z_{\eta},Y_\eta)-Z_{\eta}\frac{h}{r} \Big)\Big]= \widehat{Q}(z,y)-z\frac{h}{r}+ \mathbb{E}_{z,y}\bigg[ \int_0^{\eta}  e^{-(\rho+m) s }\mathcal{L}\Big(\widehat{Q}( Z_s,Y_s)-Z_s\frac{h}{r}\Big) ds               \bigg ],
\end{align*}
where $\mathcal{L}$ is defined in (\ref{L}).

Combining (\ref{jzy}) and (\ref{4-3}), we have 
 \begin{align}\label{hatjz}
\widehat{v}(z,y)& =\sup_{  \eta \in \mathcal{S}
} \mathbb{E}_{z,y}\bigg[    \int_0^{\eta }  e^{- (\rho+ m) s} \Big(\widehat{u}(Z_s)+Z_sY_s\Big) ds  \nonumber + \int_{0}^{\eta} e^{- (\rho+ m) s}\mathcal{L}\Big(\widehat{Q}(Z_s,Y_s)-Z_s\frac{h}{r}\Big) ds               \bigg ] \nonumber \\&
=\sup_{ \eta \in \mathcal{S}} \mathbb{E}_{z,y}\bigg[    \int_0^{\eta}  e^{- (\rho+ m) s} \Big(Z_s h-m\,u(lB)\Big)ds  \bigg],
\end{align}
where we have used the fact that (cf.\ (\ref{LW}))
\begin{align*}
\mathcal{L}(\widehat{Q}(z,y)-z \frac{h}{r})&= \mathcal{L}\widehat{Q}( z,y) -\mathcal{L}(z \frac{h}{r})=-\widehat{u}(z)-m\,u(lB)-zy+zh.
\end{align*}

From (\ref{hatjz}) we see that $\widehat{v}(z,y)$ is independent of $y$, so that $\widehat{v}$ is the value of a one-dimensional optimal stopping problem for the process $Z$. Hence, in the following, with a slight abuse of notation, we simply write $\widehat{v}(z)$.

As usual in optimal stopping theory, we let
\begin{align*}
\mathcal{C}:=\{ z \in \mathbb{R}_+: \widehat{v}(z)>0  \},\quad
\mathcal{R}:=\{ z \in \mathbb{R}_+: \widehat{v}(z)=0  \}
\end{align*}
be the so-called continuation (waiting) and stopping (purchasing) regions, respectively. We denote by $\partial \mathcal{C}$ the boundary of the set $\mathcal{C}.$
\begin{sloppypar}
Since, for any stopping time $\eta$, the mapping $z \rightarrow   \mathbb{E}_{z}[\int^\eta_0   z h e^{-rs-\frac{1}{2}\theta^2 s-\theta W_s}ds -\int^\eta_0 e^{-(\rho+m)s} m\, u(lB)ds  ]$ is continuous, then $\widehat{v}$ is lower semicontinuous on $\mathbb{R}_+$. Hence, $\mathcal{C}$ is open, $\mathcal{R}$ is closed, and introducing the stopping time 
\begin{align*}
\eta^*:= \inf\{t \geq 0 :Z_t \in \mathcal{R}\}, \quad \mathbb{P}_{z}\text{-}a.s.,
\end{align*}
with $\inf \emptyset = + \infty$, one has that $\eta^*$ is optimal for $\widehat{v}(z)$ (see, e.g.,\ Corollary I.2.9 in \cite{peskir2006optimal}) for any $z>0$.  

We now derive some preliminary properties of $\widehat{v}$ that will lead the ``guess-and-verify" analysis of Section \ref{3-2-5} below.
\begin{pro}\label{finite}
The function $\widehat{v}$ is such that $0\leq \widehat{v}(z)\leq \frac{zh}{r}$ for all $z \in \mathbb{R}_+$.
\end{pro}
\begin{proof}
The proof is given in Appendix \ref{prooffinite}.
\end{proof}
\end{sloppypar}

From $\widehat{v}$ as in the (\ref{hatjz}), the next monotonicity of $\widehat{v}$ follows.
\begin{pro}\label{monotonicity}
$z \mapsto \widehat{v}(z) $ is non-decreasing.
\end{pro}

Thanks to the previous monotonicity it is easy to see that the boundary $\partial \mathcal{C}$ can be represented by a constant $b \geq 0$. 
\begin{lemma}\label{lemmab}
Introduce the free boundary $b:=\sup \{z>0: \widehat{v}(z) \leq 0\}$ (with the convention $\sup \emptyset =0$). Then one has 
\begin{align*}
\mathcal{R}=\{ z \in \mathbb{R}_+ : 0<z \leq b \}.
\end{align*}

\end{lemma}


\subsubsection{Characterization of the free boundary and of the value function}\label{3-2-5}

We notice that the process $Z$ in (\ref{Z}) {is a strong} Markov diffusion with the infinitesimal generator given by 
\begin{align}\label{operator}
\mathbb{L}_Z = \frac{1}{2}\theta^2 z^2 \frac{\partial^2}{\partial z^2}+ (\rho-r+m)z \frac{\partial}{\partial z}.
\end{align}

By the dynamic programming principle, we expect that $\widehat{v}$ identifies with a suitable solution $\widehat{w}$ to the Hamilton-Jaccobi-Bellman (HJB) equation
\begin{align}\label{hjb1}
\max \{ [\mathbb{L}_Z-(\rho+m)]\widehat{w}+hz-m\,u(lB), -\widehat{w}  \} =0, \quad z>0.
\end{align}

In the following, we will use a classical ``guess and verify'' approach to provide explicit characterizations for both the value function $\widehat{v}(z)$ and the optimal purchasing time. Given the fact that $\widehat{v}$ is nondecreasing by Proposition \ref{monotonicity} and there exists $b$ separating $\mathcal{C}$ and $\mathcal{R}$ by Lemma \ref{lemmab}, we transform (\ref{hjb1}) into the free boundary problem 
\begin{equation}\label{guess}
	\left \{
	\begin{aligned}
		\Big(\mathbb{L}_Z-(\rho+m)\Big)\widehat{w}(z)=-hz+m\,u(lB),\quad \forall &z\in (b,+\infty), \\
		\widehat{w}(z)=0, \quad \forall &z\in(0,b].
			\end{aligned}
	\right.
\end{equation}
Solving (\ref{guess}), we have 
\begin{equation}\label{solution}
	\widehat{w}(z)=\left \{
	\begin{aligned}
		C_1 z^{\alpha_1}+C_2 z^{\alpha_2}+\frac{h}{r}z-\frac{m\,u(lB)}{\rho+m},\quad \forall &z\in (b,+\infty), \\
		0, \quad \forall &z\in(0,b],
			\end{aligned}
	\right.
\end{equation}
where $C_1$ and $C_2$ are undetermined constants and  $\alpha_1<0< 1<\alpha_2$ are the real roots of the algebraic equation
\begin{align}\label{root}
\frac{1}{2}\theta^2 \alpha^2+(\rho-r+m-\frac{1}{2}\theta^2)\alpha-(\rho+m)=0.
\end{align}
To specify the parameters $C_1,C_2$ and $b$, we first observe that 
since $\widehat{v}$  diverges at most linearly by Proposition \ref{finite} and $\alpha_2>1$, we thus set $C_2=0$. Then we appeal to the so-called ``smooth fit principle'', which dictate that the candidate value function $\widehat{w}(z)$ should be $C^1$ in $b$. These conditions give rise to the system of equations
\begin{equation}\label{system}
	\left\{
	\begin{aligned}
		C_1 b^{\alpha_1} +\frac{h b}{r}-\frac{m\,u(lB)}{\rho+m}=0,\\
C_1 \alpha_1 b^{\alpha_1-1}+\frac{h}{r}=0.	\end{aligned}
	\right.
\end{equation}
Solving (\ref{system}) one gets 
\begin{align}\label{b}
b= \frac{m\, u(lB)r \alpha_1}{(\rho+m)h(\alpha_1-1)} \quad \text{and} \quad C_1=-\frac{h}{r \alpha_1}\Big( \frac{m\, u(lB)r \alpha_1}{(\rho+m)h(\alpha_1-1)}\Big)^{1-\alpha_1}. 
\end{align}
Therefore, together with (\ref{solution}) and (\ref{b}), we conclude that 
\begin{equation}\label{vz}
\widehat{w}(z)= \left\{
\begin{aligned}
-\frac{h}{r \alpha_1}\Big( \frac{m\, u(lB)r \alpha_1}{(\rho+m)h(\alpha_1-1)}\Big)^{1-\alpha_1} z^{\alpha_1}+ \frac{h}{r}z-\frac{m\,u(lB)}{\rho+m}, \quad \text{if} \ & z>b,\\
0, \quad \text{if} \ & 0<z\leq b.
\end{aligned}
\right.
\end{equation}

\begin{theorem}\label{verification}
{Suppose that $\rho+m>r$.} The function $\widehat{w}$ given by (\ref{vz}) identifies with the dual value function $\widehat{v}$ in (\ref{hatjz}). Moreover, the optimal stopping time takes the form 
\begin{align}\label{etab}
	\eta^*=\eta^*(z;b):=\inf \{t \geq 0: Z_t^z \leq b\},
\end{align}
with $b$ as in (\ref{b}). 
\end{theorem}

\begin{proof}
	The proof is given in Appendix \ref{proofverification}.
\end{proof}
We interpret $Z_t$ to be the agent's shadow price process for the optimization problem. Economically speaking, the $Z$ process and the wealth process $X$ are closely connected. The optimal timing problem so far has been effectively addressed by examining the dual problem associated with the $Z$ process. Specifically, life insurance is recommended for purchase when the shadow price falls below the boundary $b$. This recommendation aligns with the notion that a lower shadow price corresponds to a healthier economy. Therefore, acquiring life insurance is advisable in times when the economy is deemed sufficiently robust.   

Due to Theorem \ref{verification}, (\ref{4-3}) and (\ref{vz}) we then have the following immediate corollary. 
\begin{corollary}\label{corollary}
 The function $v \in C^{2,2}(\mathcal{C}) \cap C^{1,2}(\mathbb{R}^2_+)$ and satisfies the following equation
\begin{align*}
0 = \max \Big\{  \widehat{Q}-\frac{h}{r}z-v, \frac{1}{2}\theta^2 z^2 v_{zz}+(\rho-r+m)zv_z+\frac{1}{2}\sigma_y^2 y^2v_{yy}+ \mu_y yv_y- \theta  \sigma_y zyv_{zy}\nonumber \\+\widehat{u}(z)+zy-(\rho+m)v      \Big  \}.
\end{align*}

\end{corollary}

\subsection{Optimal strategies in terms of the primal variables}\label{sec3-3}
In the previous section, we studied the properties of the dual value function $v(z,y)$ and used $(z,y)$, where $z$ denotes dual state variable  
and $y$ denotes labour income, as the coordinate system for the study. In this section, we will establish the duality theorem and derive the optimal strategies in the original coordinate system $(x,y)$, where $x$ denotes the wealth of the agent.
\begin{pro}\label{convex}
The function $v$ in (\ref{jzy}) is strictly convex with respect to $z$. Moreover, $v_z(z,y)$ satisfies the following limiting behaviors: 
\begin{align*}
	\lim_{z\to 0+} - v_z(z,y)= +\infty \quad \text{and} \quad \lim_{z\to +\infty} - v_z(z,y)= -\frac{y}{\kappa}.
	\end{align*}
\end{pro}

\begin{proof}
    From (\ref{W2}), it is easy to check that $\widehat{Q}$ is strictly convex with respect to $z$. From (\ref{hatjz}), $\widehat{v}$ is convex with respect to $z$ since it is the supremum of linear functions. Therefore,  the function $v$ in (\ref{jzy}) is strictly convex with respect to $z$ due to (\ref{4-3}). Moreover, it is straightforward to obtain the above limiting behaviors from (\ref{4-3}), (\ref{vz}) and (\ref{W1}). 
\end{proof}

We now provide the duality theorem and optimal strategies.
\begin{theorem}\label{dualrelation2}
	For given $(x,y) \in \mathcal{O}$, the following duality relations hold:
\begin{align*}
V(x,y) = \inf_{z>0}[{v}(z,y)+zx], \quad v(z,y) = \sup_{x>-\frac{y}{\kappa}} [V(x,y)-zx].
\end{align*}
Moreover, optimal consumption process $c^*_t= (Z_t^*)^{-\frac{1}{\gamma}}$, optimal portfolio process $\pi^*_t= \frac{\theta v_{zz} Z_t^*-\sigma_y v_{zy} Y_t}{\sigma}$ and the corresponding optimal wealth process $X_t^*=-v_z(Z^*_t,Y_t)$, where $Z_t^*$ is the solution to Equation (\ref{Z}) with the initial condition $z^*$ satisfying  $x=-v_z(z^*,y)$.
  \end{theorem}

\begin{proof}
The proof is given in Appendix \ref{proofdualrelation2}.
\end{proof}



From Theorem \ref{dualrelation2}, for any $(x,y) \in \mathcal{O}$, we know that $V(x,y) = \inf_{z>0}[{v}(z,y)+z x]$. Since $z \mapsto {v}(z,y)+zx$ is strictly convex (cf.\ Proposition \ref{convex}), then there exists a unique solution $z^*(x,y)>0$ such that 
\begin{align}\label{5-1}
V(x,y)=v(z^*(x,y),y)+z^*(x,y)x,
\end{align} 
where $z^*(x,y):=\mathcal{I}^{v}(-x,y)$ and $\mathcal{I}^v$ is the inverse function of $v_z$.  Moreover, $z^* \in C(\mathcal{O})$ due to $v \in C^{2,2}(\mathcal{C}) \cap C^{1,2}(\mathbb{R}^2_+)$ from Corollary \ref{corollary} and $z^*(x,y)$ is strictly decreasing with respect to $x$, which is a bijection form. Hence, for any $y\in \mathbb{R}_+$, $z^*(\cdot,y)$ has an inverse function $x^*(\cdot,y)$, which is continuous, strictly decreasing, and maps $\mathbb{R}_+$ to $(-\frac{y}{\kappa},\infty)$.

Let us now define
\begin{equation}\label{5-3}
\left\{
\begin{aligned}
{\widehat{b}(y)}&:=x^*(b,y),\\
\widehat{\mathcal{C}}&:=\{(x,y) \in \mathcal{O}: z^*(x,y) \in \mathcal{C} \},\\
\widehat{\mathcal{R}} &:=\{(x,y) \in  \mathcal{O}: z^*(x,y) \in \mathcal{R} \}.
\end{aligned}
\right.
\end{equation}

Then, by Lemma \ref{lemmab} we have 

\begin{align}\label{5-4}
\widehat{\mathcal{C}}=\{(x,y) \in \mathcal{O}: -\frac{y}{\kappa}<x <{\widehat{b}(y)} \}, \quad \widehat{\mathcal{R}} =\{(x,y) \in \mathcal{O}:  x\geq {\widehat{b}(y)} \}.
\end{align}

%



We now state the explicit expressions of the value function and optimal policies in terms of the primal variables. 
\begin{theorem}\label{optimal}
The value function $V$ in (\ref{vxy}) is given by 
\begin{equation*}
V(x,y)= \left\{
\begin{aligned}
C_1  (z^*)^{\alpha_1}(x,y)+\frac{\gamma (z^*)^{\frac{\gamma-1}{\gamma}}(x,y)}{(1-\gamma)K} + (\frac{y}{\kappa}+x)z^*(x,y), \quad \text{if} \ &-\frac{y}{\kappa}< x< {\widehat{b}(y)},\\
\frac{(x-\frac{h}{r}+\frac{y}{\kappa})^{1-\gamma}K^{-\gamma}}{1-\gamma} + \frac{m\,u(lB)}{\rho+m},\quad \text{if} \ & x \geq {\widehat{b}(y)}.
\end{aligned}
\right.
\end{equation*}
The optimal policies are $(c^*_t,\pi^*_t, \eta^*)$,  with $c_t^*=c^*(X^*_t,Y_t), \pi^*_t=\pi^*(X^*_t,Y_t)$, 
where the feedback rules $c^*$ and $\pi^*$ are such that
\begin{equation*}
c^*(x,y):= \left\{
\begin{aligned}
(z^*)^{-\frac{1}{\gamma}}(x,y), \quad \text{if} \ & -\frac{y}{\kappa}<x< {\widehat{b}(y)},\\
K(x-\frac{h}{r}+\frac{y}{\kappa}), \quad \text{if} \ & x \geq {\widehat{b}(y)},
\end{aligned}
\right.
\end{equation*}
and
\begin{equation*}
\pi^*(x,y):= \left\{
\begin{aligned}
\frac{\theta \Big[C_1 \alpha_1(\alpha_1-1)(z^*)^{\alpha_1-1}(x,y)+\frac{(z^*)^{-\frac{1}{\gamma}}(x,y)}{K \gamma}\Big]- \frac{\sigma_y y}{\kappa}}{\sigma}, \quad \text{if} \ & -\frac{y}{\kappa}<x< {\widehat{b}(y)},\\
\frac{\theta  (x-\frac{h}{r}+\frac{y}{\kappa}) \frac{1}{ \gamma}-  \frac{\sigma_y y}{\kappa}}{\sigma}, \quad \text{if} \ & x \geq {\widehat{b}(y)},
\end{aligned}
\right.
\end{equation*}
with $z^*{(x,y)}$ satisfying 
\begin{align*}
 C_1 \alpha_1 (z^*{(x,y)})^{\alpha_1-1}-(z^*{(x,y)})^{-\frac{1}{\gamma}}\frac{1}{K}+\frac{y}{\kappa}+x=0,
\end{align*}
for $C_1$ given in (\ref{b}). 
Furthermore,
\begin{align*}
	\eta^*(x,y)=\inf \{t\geq 0:X^{*x}_t \geq {\widehat{b}(Y_t^y)}\},
\end{align*}
with ${\widehat{b}(y)}=\frac{b^{-\frac{1}{\gamma}}}{K}-\frac{y}{\kappa}+\frac{h}{r}$, and the optimal wealth process $X^*$ is such that $X^{*}_t=-v_z(Z^*_t,Y_t)$, where $Z^*_t$ is the solution to Equation (\ref{Z}) with the initial condition $Z_0=z^*$, and $z^*:=z^*(x,y)$ is the solution to the equation $v_z(z,y)+x=0$, with $x$ being the initial wealth at time $0$, and
\begin{equation*}
v_{z}(z,y)= \left\{
\begin{aligned}
C_1 \alpha_1 z^{\alpha_1-1}-\frac{z^{-\frac{1}{\gamma}}}{K}+\frac{y}{\kappa}, \quad \text{if} \ & z>b,\\
 -\frac{z^{-\frac{1}{\gamma}}}{K}+\frac{y}{\kappa}-\frac{h}{r},  \quad \text{if} \ & 0<z\leq b.
\end{aligned}
\right.
\end{equation*}
\end{theorem}

\begin{proof}
	The proof is given in Appendix \ref{proofoptimal}.
\end{proof}

\subsection{Numerical illustrations}\label{sec3num}
In this section, we provide numerical illustrations of the optimal strategies and of the value functions derived in Theorem \ref{optimal}. Moreover, we investigate the sensitivities of the optimal purchasing wealth threshold on relevant parameters. The numerics was performed using Mathematica 13.1. 

\subsubsection{Parameters}
We fix the parameters for the financial market and labor income process similarly to \cite{dybvig2010lifetime}: risky-asset return $\mu=5 \%$, risky-asset volatility $ \sigma=22 \% $, risk-free interest rate $r=1 \%,$ {initial labor income} $y=1$, income mean growth rate $\mu_y=1 \%$, income volatility $\sigma_y=10\%$, {weight parameter $l=0.5$}.  For the individual, we assume that she is currently $x_0=25$ years old and decides for her optimal purchasing time for life insurance. Following \cite{chen2021retirement} we take for the individual's force of mortality $m=0.0175$. For the preferences, we assume that the subjective discount rate $\rho$ is equal to $r$, and $\gamma=0.8$. These basic parameters are collected in Table \ref{tab1}. It is worth noting that all the parameters we used in numerical examples satisfy the Assumption \ref{assume2}. 
\begin{table}[htbp]
\caption{Basic parameters set in the numerical illustrations }
\label{tab1}
\centering
	\begin{tabular}{cccccccccc}
		\toprule  
		 $\mu$	& $\sigma$ &$r$ & $\rho$   &$\gamma$ & 
		  $\mu_y$ & $\sigma_y$ &$l$ & $m$ & $B$ \\ 
		  \midrule
	0.05	&0.22 &0.01 &0.01& 0.8 &0.01&0.1 & 0.5 &0.0175&5\\
	
		\bottomrule
	\end{tabular}
\end{table}




\subsubsection{Sensitivity analysis of the optimal boundary}
In this section, we study the sensitivity of the optimal purchasing wealth threshold  with respect to model's parameters and provide the consequent economic implications. 

In Figure \ref{p5} we can observe the sensitivity of the optimal purchasing wealth threshold with respect to the initial labor income. Since an increase in $y$ implies higher human capital (the present value of future labor income), the agent is more likely to buy life insurance earlier. Figure \ref{p6} shows that if the predetermined bequest amount $B$ is larger, the agent delays her decision to purchase life insurance. This is an intuitive result, since a larger $B$ means a larger premium $h= mB$.  
\begin{figure}[htbp]
		\setlength{\abovecaptionskip}{0pt}
	\setlength{\belowcaptionskip}{5pt}
	\begin{minipage}[b]{0.47\textwidth}
		\centering
		\includegraphics[width=\textwidth]{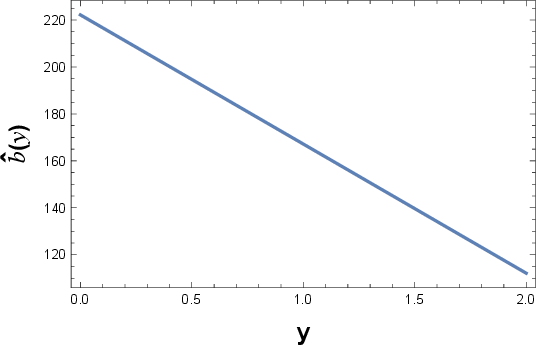}
		\caption{Values of ${\widehat{b}(y)}$ when varying $y$}
		\label{p5}
	\end{minipage}
	\hfill
		\begin{minipage}[b]{0.47\textwidth}
		\centering
		\includegraphics[width=\textwidth]{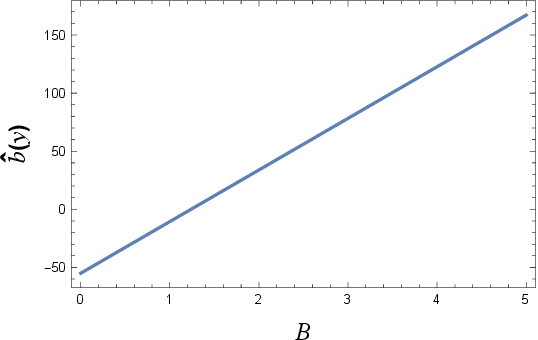}
		\caption{Values of ${\widehat{b}(y)}$ when varying $B$ at $y=1$}
		\label{p6}
	\end{minipage}
\end{figure}


Figure \ref{p7} shows that if the risk aversion level $\gamma$ is larger, the agent is more likely to buy life insurance earlier.  As a matter of fact, the incentive of life insurance is to reduce the longevity risk and bequest motive, thus agents who are more risk averse are more willing to buy insurance earlier. In particular, when the risk aversion level $\gamma$ goes to $1$, the optimal purchasing boundary converges to $\frac{h}{r}-\frac{y}{\kappa}$ (the dotted line). It implies that the agent should buy life insurance immediately, which is consistent to the result in Theorem \ref{buy}. Figure \ref{p8} shows the effect of a change in $l$ on the boundary ${\widehat{b}(y)}$. It is clear that if $l$ is larger, the agent assigns higher utility value to bequest, so that the agent will choose buy life insurance earlier.
\begin{figure}[htbp]
		\setlength{\abovecaptionskip}{0pt}
	\setlength{\belowcaptionskip}{5pt}
	\begin{minipage}[b]{0.47\textwidth}
		\centering
		\includegraphics[width=\textwidth]{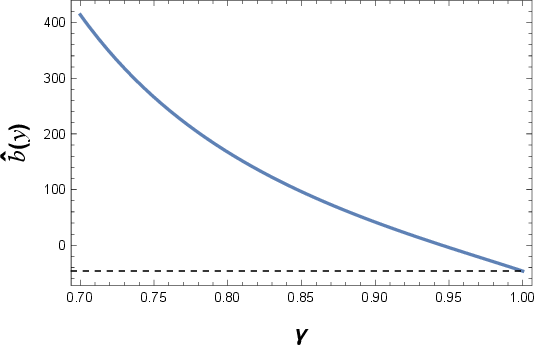}
		\caption{Values of ${\widehat{b}(y)}$ when varying $\gamma$ at $y=1$}
		\label{p7}
	\end{minipage}
	\hfill
		\begin{minipage}[b]{0.47\textwidth}
		\centering
		\includegraphics[width=\textwidth]{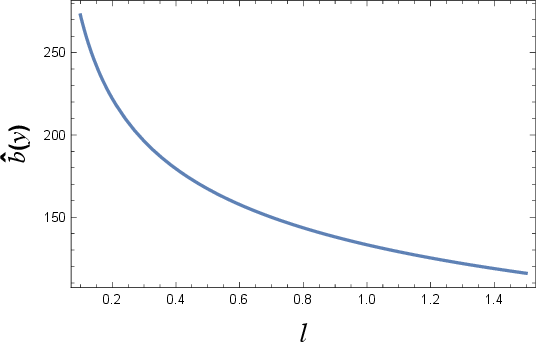}
		\caption{Values of ${\widehat{b}(y)}$ when varying $l$ at $y=1$}
		\label{p8}
	\end{minipage}
\end{figure}

\subsubsection{Optimal investment and consumption strategies}
In this section, we illustrate the ratios of optimal strategies at time $0$ to initial wealth. Figure \ref{p3} shows a jump in the ratio of $\pi^*(x,y)/x$ in correspondence to the critical wealth level ${\widehat{b}(y)}$. Actually, this effect shares similarities to the so-called ``saving for retirement", where the optimal portfolio has a jump down at retirement (cf.\ \cite{dybvig2010lifetime}).

Figure $\ref{p4}$ shows the ratio of optimal consumption at time $0$ to initial wealth. Interestingly, unlike the optimal investment strategy, the ratio of optimal consumption to wealth is smooth and there is no jump at the purchasing boundary because the marginal utility per unit of consumption does not change after buying life insurance. 

\begin{figure}[htbp]
		\setlength{\abovecaptionskip}{0pt}
	\setlength{\belowcaptionskip}{5pt}
	\begin{minipage}[b]{0.47\textwidth}
		\centering
		\includegraphics[width=\textwidth]{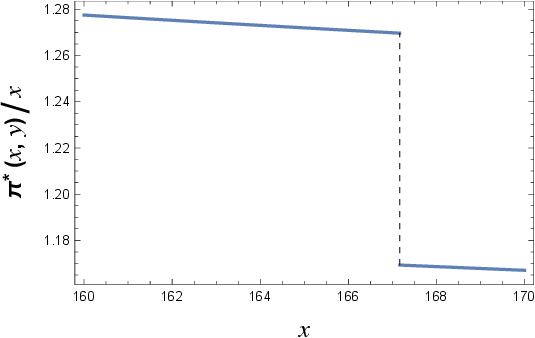}
		\caption{Optimal portfolio feedback strategy at $y=1$}
		\label{p3}
	\end{minipage}
	\hfill
		\begin{minipage}[b]{0.47\textwidth}
		\centering
		\includegraphics[width=\textwidth]{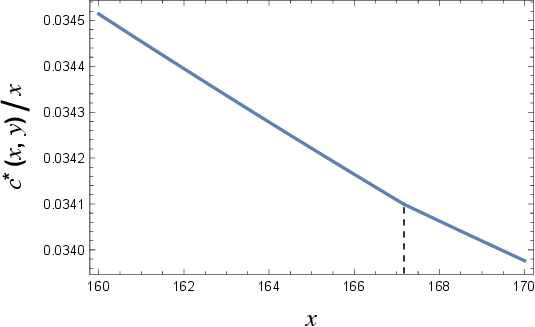}
		\caption{Optimal consumption feedback plan at $y=1$}
		\label{p4}
	\end{minipage}
\end{figure}

\section{A controlled bequest amount} \label{sec:conbeq}

\subsection{Problem formulation}
In this section, we assume the agent can choose how much money she plans to bequeath at death, i.e., $\{B_t, t \geq 0\}$ is an $\mathcal{F}_t$-adapted control variable. Once purchased, the amount of bequest will be fixed. Correspondingly, when the agent chooses a target amount $B_t$ at purchasing time $t$, the premium $h_t$ in (\ref{premium}) is such that $h_t=mB_t$. 

From (\ref{wealth}), the agent's wealth $X^B:= \{ X^{c,\pi, B,\eta}_t, t\geq 0 \}$ now evolves as 
\begin{align}\label{stateeq}
d X^{c,\pi, B,\eta}_t =  [\pi_t(\mu-r)+rX^{c,\pi,B,\eta}_t-c_t-mB_t \mathds{1}_{\{t \geq \eta \}}+Y_t]dt +\pi_t \sigma dW_t, \quad X_0^{c,\pi, B,\eta}= x.
\end{align}
In the following, we shall simply write $X^B$ to denote $X^{c,\pi,B,\eta}$, where needed. Then, the set of admissible strategies $\mathcal{A}^B(x,y)$  is as follows. 

\begin{definition}\label{admissiblecontrol2}
Let $(x,y) \in \mathcal{O}$ be given and fixed. The triplet of choices $(c,\pi, B, \eta)$ is called an \textbf{admissible strategy} for $(x,y)$, and we write $(c,\pi, B,\eta) \in \mathcal{A}^B(x,y)$, if it satisfies the following conditions: 
\begin{enumerate}[label=(\roman*)]
	\item  $c$ and $\pi$ are progressively measurable with respect to $\mathbb{F}$, $\eta \in \mathcal{S}$;
	\item $c_t \geq 0$ for all $t \geq 0$ and $\int_0^t(c_s+|\pi_s|^2)ds <\infty$ for all $t\geq 0$ $ \ \mathbb{P}$-a.s.; 
\item $B_t  =0$ for all $0\leq t <\eta$ and $B_t=B_\eta \geq 0$ for any $t \geq \eta$, where $B_\eta$ is an $\mathcal{F}_\eta$-measurable random variable; 
	\item $X^{c,\pi,B, \eta}_t +g_t > \frac{m B_t}{r}\mathds{1}_{\{t\geq  \eta\}}$ for all $t\geq 0$, where $g_t$ is defined in (\ref{gt}).
\end{enumerate}

\end{definition}
The term $\frac{m B_t}{r}:= \mathbb{E} [\int^{\infty}_t \frac{\xi_s}{\xi_t} m B_t ds |\mathcal{F}_t ]$ in Condition (iv) is the present value of the future premium payment of the agent, under the assumption that the agent is always alive. 

Similar to (\ref{vxy}), and from (\ref{objective2}) the agent aims at determining
\begin{align}\label{6-1}
V^B(x,y):=\sup_{(c, \pi,B,\eta) \in \mathcal{A}^B(x,y) }\mathbb{E}_{x,y}\bigg[\int^\eta_0e^{-(\rho+m) t}\Big(u(c_t)+m\,u(0)\Big)dt \nonumber\\+\int^\infty_\eta e^{-(\rho+m) t} \Big(u(c_t)+m\,u(lB_\eta)\Big) dt  \bigg].
\end{align}
In this section, we will focus on (\ref{6-1}). Similar to Theorem \ref{buy}, we also have the following immediate result.  
\begin{theorem}\label{the4.1}
	When $\gamma>1$, the optimal purchasing time is $\eta^*=0$.
\end{theorem}
Hence, in the following, we focus on the case $0<\gamma<1$.

\subsection{Solution to the problem for $\gamma \in (0,1)$}\label{sec4-2}
In this subsection, we determine the solution to (\ref{6-1}) when $0<\gamma<1$, by employing a duality approach, which is analogous to the methods used in Section 3. First, we conduct successive transformations (cf.\ Subsections \ref{sec4-2-1}, \ref{sec4-2-2} and \ref{sec4-2-3}) that connect the original stochastic control-stopping problem (with value function $V^B$) into its dual problem (with value function $v^B$). Then we study the reduced-version dual stopping problem (with value function $\widehat{v}^B$) and give our optimal purchasing time and optimal bequest amount of this section, see Theorem \ref{main}. 

\subsubsection{The static budget constraint}\label{sec4-2-1}

We now write down the static budget constraint by using the well-known method developed by \cite{karatzas1987optimal} and \cite{cox1989optimal}:
\begin{align}\label{6-2}
		\mathbb{E}_{x,y}\big[{\xi_{t\wedge \eta}}X_{t\wedge \eta}\big]+\mathbb{E}_{x,y}\bigg[ \int_0^{t\wedge \eta} \xi_{u} c_u du                         \bigg]  \leq x+ \mathbb{E}_{x,y}\bigg[ \int_0^{t\wedge \eta} \xi_{u} Y_u du \Bigg], \quad  \ t \geq 0,
	\end{align}
and
\begin{align}\label{6-3}
\mathbb{E}_{x,y}\big[{\xi_{t \vee \eta}}X_{t \vee \eta}\big]+\mathbb{E}_{x,y}\bigg[ \int_{\eta}^{t\vee \eta} \xi_u \, \Big(c_u+mB_0\Big) \, d u   \bigg]   \leq X_\eta+ \mathbb{E}_{x,y}\bigg[ \int_\eta^{t \vee \eta} \xi_{u} Y_u du \Bigg], \  \ t \geq 0.
\end{align}

\subsubsection{The optimization problem after purchasing life insurance}\label{sec4-2-2}
 
In this subsection, we shall look into the agent's optimization problem after the purchase of life insurance. Apart from determining consumption and portfolio choices, the agent is also required to determine the optimal bequest amount $B^*_{\eta}$ at purchasing time $\eta$, where $\eta=0$ in this subsection.

Letting $\mathcal{A}^B_0(x,y):= \{ (c,\pi,B): (c,\pi,B,0) \in \mathcal{A}^B(x,y)\},$ where the subscript $0$ indicates that the purchasing time $\eta$ is equal to $0$, the agent's value function after purchasing life insurance then reads as 
\begin{align}\label{6-4}
    \widehat{V}^B(x,y):= \sup_{(c,\pi,B)\in \mathcal{A}_0^B(x,y)} \mathbb{E}_{x,y	}\bigg[\int^\infty_0 e^{-(\rho+m) s } \Big(u(c_s)+m \, u(lB_0)\Big) ds \bigg].
    \end{align}
Recalling that $\xi_{s}=e^{-rs-\theta W_s -\frac{1}{2}\theta^2 s}$, and for any pair $(c,\pi,B)\in \mathcal{A}_0^B(x,y)$ with a Lagrange multiplier $z>0$, we have
\begin{align}\label{6-5}
 \mathbb{E}&_{x,y} \bigg[    \int_0^{\infty}  e^{-(\rho+m) s } \Big(u(c_s)+m  \, u(lB_0)\Big)  ds  \bigg] \nonumber \\ \leq &\
 \mathbb{E}_{x,y} \bigg[    \int_0^{\infty} e^{-(\rho+m) s } \Big(u(c_s)+m  \, u(lB_0)\Big)  ds  \bigg]   - z      \mathbb{E}_{x,y}\bigg[  \int_0^{\infty} \xi_{s} \Big(c_s+mB_0-Y_s\Big) ds                         \bigg]  +z x                  \nonumber               \\
=&\   \mathbb{E}_{x,y}\bigg[    \int_0^{\infty}  e^{-(\rho+m) s } \bigg(u(c_s) -z P_sc_s+z P_sY_s  \bigg) ds  \bigg]\nonumber \\
&+\mathbb{E}_{x,y}\bigg[   \int_0^{\infty }  e^{-(\rho+m) s } \Big(m \, u(l B_0)-z P_s m B_0\Big) ds\bigg]+ z x  \nonumber \\
\leq &\ \mathbb{E}_{x,y}\bigg[    \int_0^{\infty } e^{-(\rho+m) s }\Big(\widehat{u}(zP_s)+ zP_sY_s\Big) ds  \bigg]+ z x+\bar{u}(z),
\end{align}
where the first inequality results from the budget constraint stated in \eqref{6-3}, $P$ and $\widehat{u}$ are defined in (\ref{pt}), and
\begin{align}\label{6-7}
 \quad \bar{u}(z)&:= \sup_{B_0 \geq 0}\mathbb{E}_{x,y}\bigg[   \int_0^{\infty }  e^{-(\rho+m) s } \Big(m \, u(l B_0)-z P_s m B_0\Big) ds\bigg]\nonumber\\
 &= \sup_{B_0 \geq 0} \Big[\frac{m \, u(lB_0)}{\rho+m}-z \frac{mB_0}{r}\Big]=m l^{\frac{1-\gamma}{\gamma}} \frac{\gamma}{1-\gamma}r^{\frac{1-\gamma}{\gamma}}(\rho+m)^{-\frac{1}{\gamma}}z^{\frac{\gamma-1}{\gamma}},
  \end{align}
with the optimizing bequest amount being 
\begin{align}\label{bequest}
	B_0^*:=\Big[\frac{z(\rho+m)}{r}\Big]^{-\frac{1}{\gamma}}l^{\frac{1-\gamma}{\gamma}} .
	\end{align}
 Since the amount of bequest will be fixed once purchased, the (candidate) optimal bequest amount $B^*_0$ will just depend on the initial state $z$. This means the agent can choose the optimal bequest amount based on their initial state $z$ at the time of purchase ($\eta=0$ in this subsection). 

Then we set
\begin{align}\label{6-8}
\widehat{Q}^B(z,y):=  \mathbb{E}_{z,y}\bigg[    \int_0^{\infty }  e^{- (\rho+ m) s} \Big[\widehat{u}(Z_s)+Z_sY_s\Big] ds  \bigg],
\end{align}
where $Z$ is given in (\ref{Z}).

\begin{pro}\label{pro4.1}
One has $\widehat{Q}^B \in  C^{2,2}(\mathbb{R}^2_+).$ Moreover, $\widehat{Q}^B$  satisfies 
\begin{align}\label{6-9}
-{\mathcal{L}}\widehat{Q}^B=\widehat{u}+zy, \ \text{on} \ \mathbb{R}^2_+, 
\end{align}
where $\mathcal{L}$ is defined in (\ref{L}).

\end{pro}
\begin{proof}

The proof is given in \ref{proofpro4.1}.

\end{proof}

Similar to Theorem \ref{dualrelation1}, it is possible to relate $\widehat{V}^B$ to $\widehat{Q}^B$ through the following duality relation.
\begin{theorem}\label{thedual}
The following dual relations hold:
\begin{align*}
\widehat{V}^B(x,y) = \inf_{z >0} [ \widehat{Q}^B(z,y)+zx+\bar{u}(z) ], \quad \widehat{Q}^B(z,y) =  \sup_{x> -\frac{y}{\kappa}+\frac{m B_0}{r}}[\widehat{V}^B(x,y)-zx-\bar{u}(z)].
\end{align*}
\end{theorem}

\subsubsection{The dual optimal stopping problem}\label{sec4-2-3}
Remember that we focus on $0<\gamma<1$ due to Theorem \ref{the4.1}, that is $u(0)=0$. Now, for any $(x,y) \in \mathcal{O}$ and Lagrange multiplier $z>0$, from (\ref{6-1}) and the budget constraint (\ref{6-2}), recalling that $P_s$ as in (\ref{pt}), we have by the strong Markov property
\begin{align*}
  \mathbb{E}&_{x,y}\bigg[\int^\eta_0 e^{-(\rho+m)t }u(c_t)dt +\int^\infty_\eta e^{-(\rho+m)t } \Big(u(c_t)+m\,u(lB_\eta)\Big) dt  \bigg]
 \nonumber \\ \leq &\
\sup_{(c,\pi,B,\eta) \in \mathcal{A}^B(x,y)} \mathbb{E}_{x,y}  \bigg[\int^\eta_0 e^{-(\rho+m) t }u(c_t) dt + e^{-(\rho+m) \eta }\widehat{V}^B(X_\eta,Y_\eta)  \bigg] \nonumber \\& - z      \mathbb{E}_{x,y} \bigg[   \xi_{\eta}X_{\eta}+ \int_0^{\eta} \xi_{t}\Big(c_t-Y_t\Big) dt                         \bigg]  +z x                        \nonumber       \\
=&\   \sup_{(c,\pi,B,\eta) \in \mathcal{A}^B(x,y)} \mathbb{E}_{x,y}\bigg[    \int_0^{\eta}   e^{-(\rho+m) s } \bigg(u(c_s) -zP_sc_s+zP_sY_s\bigg) ds  \nonumber \\&+ e^{-(\rho+m) \eta } \widehat{V}^B(X_{\eta},Y_\eta) -  e^{-(\rho+m) \eta }zP_{\eta}X_{\eta}      \bigg]+z x  \nonumber \\
 \leq &\  \sup_{\eta \in \mathcal{S}} \mathbb{E}_{z,y}\bigg[    \int_0^{\eta} e^{-(\rho+m) s }\Big(\widehat{u}(Z_s)+Z_sY_s\Big) ds
+  e^{-(\rho+m) \eta }\Big(\widehat{Q}^B(Z_{\eta},Y_\eta)+\bar{u}(Z_\eta)\Big) \bigg]+ zx , 
\end{align*}
where $\widehat{u}(z)= \sup_{c \geq 0}[u(c)-c z]$ and $Z_t $ is defined in (\ref{Z}). 


Hence, setting
\begin{align}\label{6-17}
v^B(z,y):=&\ \sup_{ \eta \in \mathcal{S}
} \mathbb{E}_{z,y}\bigg[    \int_0^{\eta }  e^{- (\rho+ m) s}\Big(\widehat{u}(Z_s) +Z_sY_s\Big)ds \nonumber\\&+ e^{- (\rho+ m) \eta}\Big(\widehat{Q}^B(Z_{\eta},Y_\eta)+\bar{u}(Z_\eta)\Big) \bigg],
\end{align}
we have a two-dimensional optimal stopping problem, with dynamic $(Z,Y)$ as in (\ref{Z}) and (\ref{income}). 

In the following subsections, we study (\ref{6-17}). Before doing that, similarly to Theorem \ref{dualrelation2}, we have the following theorem that establishes a dual relation between the original problem (\ref{6-1}) and the optimal stopping problem (\ref{6-17}).
\begin{theorem}
The following duality relations hold:
\begin{align*}
V^B(x,y) = \inf_{z>0}[v^B(z,y)+zx], \quad v^B(z,y) = \sup_{x>-\frac{y}{\kappa}} [V^B(x,y)-zx].
\end{align*}  
\end{theorem}

\subsubsection{Study of the dual optimal stopping problem}\label{sec4-2-4}

To study the optimal stopping problem (\ref{6-17}), it is convenient to introduce the function
\begin{align}\label{6-18}
\widehat{v}^B(z,y):={v}^B(z,y)-{\widehat{Q}^B}(z,y)-\bar{u}(z).
\end{align}

Applying It\^o's formula to $\{ e^{-(\rho+m)t }[\widehat{Q}^B(Z_t,Y_t)+\bar{u}(Z_t)] , t \in [0,\eta]  \}$, and taking conditional expectations we have 
\begin{align*}
\mathbb{E}&_{z,y}\Big[ e^{-(\rho+m) \eta }\Big(\widehat{Q}^B(Z_{\eta},Y_\eta)+\bar{u}(Z_\eta)\Big) \Big]= \widehat{Q}^B(z,y)+\bar{u}(z)\\&+ \mathbb{E}_{z,y}\bigg[ \int_0^{\eta}  e^{-(\rho+m) s }\mathcal{L}\Big(\widehat{Q}^B( Z_s,Y_s)+\bar{u}(Z_s)\Big) ds               \bigg ],
\end{align*}
where $\mathcal{L}$ is defined in (\ref{L}). Combining (\ref{6-17}) and (\ref{6-18}), we have 
 \begin{align}\label{6-19}
\widehat{v}^B(z,y)& =\sup_{  \eta \in \mathcal{S}
} \mathbb{E}_{z,y}\bigg[    \int_0^{\eta }  e^{- (\rho+ m) s} \Big(\widehat{u}(Z_s)+Z_sY_s\Big) ds  \nonumber + \int_{0}^{\eta} e^{- (\rho+ m) s}\mathcal{L}\Big(\widehat{Q}^B(Z_s,Y_s)+\bar{u}(Z_s)\Big) ds               \bigg ] \nonumber \\&
=\sup_{ \eta \in \mathcal{S}} \mathbb{E}_{z,y}\bigg[    \int_0^{\eta}  e^{- (\rho+ m) s} \Big(-K \bar{u}(Z_s)\Big)ds  \bigg],
\end{align}
where $K:= \frac{1}{\gamma}( \rho+m-r(1-\gamma)-\frac{1-\gamma}{2\gamma} \theta^2)>0$ due to Assumption \ref{assume2}, and where we have used the fact that (cf.\ (\ref{6-9}))
\begin{align*}
\mathcal{L}(\widehat{Q}^B(z,y)+\bar{u}(z))&=-\widehat{u}(z)-zy-K \bar{u}(z).
\end{align*}

From (\ref{6-19}) we see that $\widehat{v}^B(z,y)$ is independent of $y$. Hence, it is the value of a one-dimensional optimal stopping problem, and in the following, with a slight abuse of notation, we simply write $\widehat{v}^B(z)$. As $\bar{u}>0$ when $0<\gamma<1$ (cf.\ (\ref{6-7})), the following result follows immediately (cf.\ also Theorem \ref{the4.1}).

\begin{theorem}\label{main}
For any $\gamma \in (0,1) \cup (1,\infty)$, the optimal purchasing time is $\eta^*=0$; i.e., the agent should buy life insurance immediately. Moreover, the optimal bequest amount is $B^*_0=(\frac{z(\rho+m)}{r})^{-\frac{1}{\gamma}}l^{\frac{1-\gamma}{\gamma}}$.
\end{theorem}

\begin{proof}
\sloppypar
We know that when $0<\gamma<1$, $\widehat{v}^B(z)=\sup_{ \eta \in \mathcal{S}} \mathbb{E}_{z}[    \int_0^{\eta}  e^{- (\rho+ m) s} (-Km l^{\frac{1-\gamma}{\gamma}} \frac{\gamma}{1-\gamma}r^{\frac{1-\gamma}{\gamma}}(\rho+m)^{-\frac{1}{\gamma}}Z_s^{\frac{\gamma-1}{\gamma}} )ds  ]=0$ (cf.\ (\ref{6-7}) and (\ref{6-19})); hence $\widehat{v}^B(z)=0$  and $\eta^* =0$ a.s. for any $z \in \mathbb{R}_+$. Combining then Theorem \ref{the4.1} and (\ref{bequest}), we conclude. 
\end{proof}


\subsection{Optimal strategies in terms of the primal variables}

In the previous section, we studied the properties of the dual value function $v^B(z,y)$ and used $(z,y)$, where $z$ denotes the dual variable and $y$ denotes labour income, as the coordinate system for the study. In this section, we will come back to study of the value function $V^B(x,y)$ in the original coordinate system $(x,y)$, where $x$ denotes the wealth of the agent. Using arguments similar to those in Section \ref{sec3-3}, we give the explicit expressions of the value function and optimal policies in terms of the primal variables.

\begin{theorem}\label{optimal1}
	The value function $V^B$ in (\ref{6-1}) is given by 
	\begin{align*}
		V^B(x,y)=\frac{1}{ \frac{1}{K}+m(lr)^{\frac{1-\gamma}{\gamma}}(\rho+m)^{-\frac{1}{\gamma}}} \frac{(\frac{y}{\kappa}+x)^{1-\gamma} }{1-\gamma}.	\end{align*}
	The optimal policies are ($c_t^*,\pi_t^*,B^*_t,\eta^*$),  with $c_t^*=c^*(X^*_t,Y_t), \pi^*_t=\pi^*(X^*_t,Y_t)$, where 
	\begin{align*}
		c^*(x,y)&:=\frac{\frac{y}{\kappa}+x}{\frac{1}{K}+m(lr)^{\frac{1-\gamma}{\gamma}}(\rho+m)^{-\frac{1}{\gamma}} }, \\\pi^*(x,y)&:=\bigg(\frac{\frac{y}{\kappa}+x}{\frac{1}{K}+m(lr)^{\frac{1-\gamma}{\gamma}}(\rho+m)^{-\frac{1}{\gamma}} }\bigg)\frac{\theta}{\gamma K \sigma}-\frac{\sigma_y y}{\kappa \sigma},	\end{align*}
	and
	\begin{align*}
		B^*_{0}(x,y):=\bigg(\frac{\frac{y}{\kappa}+x}{\frac{1}{K}+m(lr)^{\frac{1-\gamma}{\gamma}}(\rho+m)^{-\frac{1}{\gamma}} }\bigg)\bigg(\frac{\rho+m}{r}\bigg)^{-\frac{1}{\gamma}}l^{\frac{1-\gamma}{\gamma}}.
			\end{align*}
Furthermore, $\eta^*=0$ a.s. for any $(x,y) \in \mathcal{O}$ and the optimal wealth process is such that $X^*_t= \frac{(Z^*_t)^{-\frac{1}{\gamma}}}{K}+\frac{m B^*_0}{r}-\frac{Y_t}{\kappa}$, where $Z_t^*$ is the solution to Equation (\ref{Z}) with the initial condition
	\begin{align*}
		Z_0^*=z^*(x,y):=\bigg(\frac{\frac{y}{\kappa}+x}{\frac{1}{K}+m(lr)^{\frac{1-\gamma}{\gamma}}(\rho+m)^{\frac{-1}{\gamma}} }\bigg)^{-\gamma}	\end{align*}
	 and $x$ is the initial wealth at time $0$.
	\end{theorem}
\begin{proof}
The proof is given in Appendix \ref{proofoptimal1}.
\end{proof}

\subsection{Numerical illustrations}

In this section, we provide numerical illustrations of the optimal strategies and of the value functions discussed in Theorem \ref{optimal1}. Moreover, we investigate the sensitivities of the  optimal bequest amount on relevant parameters and compare the differences between the predetermined case and the controlled case.  The basic parameters are listed in Table \ref{tab1}.


\subsubsection{Sensitivity analysis of optimal bequest amount}

In this section, we study the sensitivity of the optimal bequest amount. Figure \ref{p14} shows that if $l$ is larger, the agent is more likely to buy more life insurance because the agent aligns more utility value on bequest. Figure \ref{p13} shows the effect of a change in $\gamma$ on the optimal bequest amount. If the risk aversion level $\gamma$ is larger, the agent with more risk aversion tends to  buy more life insurance. 
\begin{figure}[htbp]
		\setlength{\abovecaptionskip}{0pt}
	\setlength{\belowcaptionskip}{5pt}
	\begin{minipage}[b]{0.47\textwidth}
		\centering
		\includegraphics[width=\textwidth]{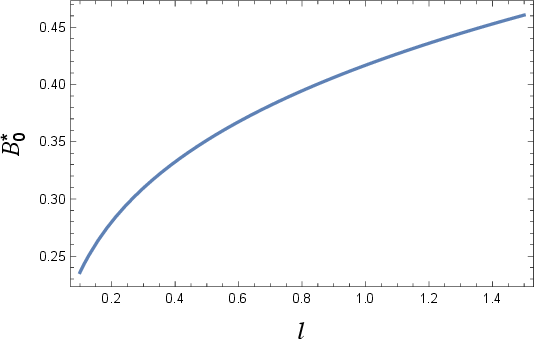}
		\caption{Values of $B^*_0$ when varying $l$ at $x=1, y=1$}
		\label{p14}
	\end{minipage}
	\hfill	\begin{minipage}[b]{0.47\textwidth}
		\centering
		\includegraphics[width=\textwidth]{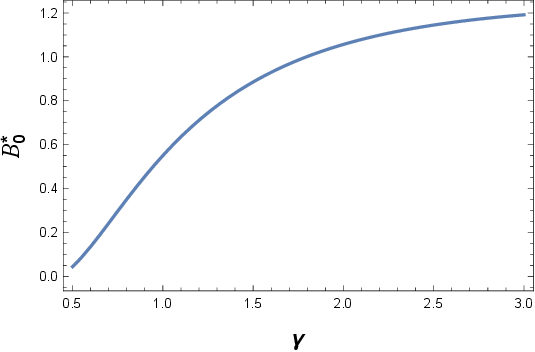}
		\caption{Values of $B^*_0$ when varying $\gamma$ at $x=1, y=1$}
		\label{p13}
	\end{minipage}
	\end{figure}

\subsubsection{Comparison of optimal strategies and value functions} Here we compare the optimal strategies at initial time $0$ and value functions in both the predetermined bequest amount and the controlled bequest amount case, see Table \ref{tab3}. We find that when an agent whose initial wealth $x=1$ and predetermined bequest amount is $B=5$, the agent will invest more money than current wealth. 
The combination of a low relative risk aversion ($\gamma=0.8$) and the anticipation of expected high future labor income has resulted in a portfolio allocation where the proportion invested in the risky asset exceeds 100\%. 
  If we choose the predetermined bequest amount $B$ equal to the optimal bequest amount $B^*_0$ derived in the controlled case, then the optimal portfolio and optimal consumption plan are the same as that in the controlled case.
\begin{table}[htbp]
\caption{Comparison of optimal strategies and value functions}
\label{tab3}
\centering
	\begin{tabular}{cccccccccc}
		\toprule  
		Predetermined:&$x$	&$B$ &$\pi^*(x,y)/x$&$c^*(x,y)/x$ &$\frac{y}{\kappa}$ \\
		  \midrule
	Values:	&1 & 5 &33.482 &1.477& 55  \\	
	\midrule
	Predetermined:&$x$	&$B$&$\pi^*(x,y)/x$&$c^*(x,y)/x$ &&  \\
	\midrule
	Values:	&1 &0.351&32.216 &1.479&   \\	
	\midrule
	Controlled:&$x$&$B^*_0(x,y)/x$	&$\pi^*(x,y)/x$&$c^*(x,y)/x$ &\\
	\midrule
	Values:	&1& 0.351 &32.216 &1.479&  \\				\bottomrule
	\end{tabular}
\end{table}

\section{{Discussions and Extensions}}\label{dissec}

{ In our initial setup, the individual decides on all consumption and life insurance matters at time $t=0$ with an initial wealth of $x$. Choosing not to buy insurance implies the policyholder does not set aside funds for inheritance, resulting in $u(0)$ in our mathematical model (see (\ref{objective})). For $\gamma>1$, this has resulted in the policyholder purchasing life insurance immediately at $t=0$ because of $u(0)=-\infty$ (cf.\ Theorems \ref{buy} and \ref{the4.1}). 

In this section, we explore a scenario where the policyholder allocates a portion of her initial wealth for inheritance. Here, the policyholder's initial wealth is divided into two parts: $x=(1-\phi)x+ \phi x:= q + x^{\prime}$, where $q:=(1-\phi)x>0 $ is directly designated for inheritance (not used for investment purpose), while $ x^{\prime}:=\phi x$ is allocated for the policyholder's personal consumption and life insurance decisions. In  case of not purchasing insurance, the money for bequest will be $q$, while in case of purchasing insurance, the money for bequest will be $q+ B_t$ if $t$ is the purchasing time. First, we consider the constant force of mortality, and then we extend the result to an age-dependent force of mortality. It is worth noting that there exists the critical wealth level regardless of the values of $\gamma$ in both cases.

\subsection{Constant force of mortality}\label{sec5-1}

To incorporate the above scenario, the agent's aim is then to maximize the following expected lifetime utility 
\begin{align}\label{eq5-1}
    \mathbb{E}\bigg[ \Big( \int^\tau_0e^{-\rho t}u(c_t)dt +e^{-\rho \tau}u(q) \Big)\mathds{1}_{\{ \eta \geq \tau\}} + \mathds{1}_{\{\eta<\tau\}}\Big( \int^\tau_0 e^{-\rho t}u(c_t)dt+e^{-\rho \tau}u(l(q+B_\tau)) \Big) \bigg],
\end{align}
where $q>0$ is a given constant. Also, similar to (\ref{objective2}), we can disentangle the market and mortality risk and write (\ref{eq5-1}) as 
\begin{align}\label{eq5-2}
   \mathbb{E}\bigg[\int^\eta_0e^{-(\rho+m) t}\Big(u(c_t)+m \, u(q) \Big)dt +\int^\infty_\eta e^{-(\rho+m) t} \Big(u(c_t)+m \,u(l(q+B_t))\Big) dt \bigg].
\end{align}
In the following, as done in Sections \ref{sec:predeq} and \ref{sec:conbeq}, we will consider a predetermined bequest amount case and a controlled bequest amount case.
\subsubsection{A predetermined bequest amount}
Using the similar arguments as in Section \ref{sec3-2}, we obtain the following dual optimal stopping problem:
 \begin{align}\label{disos}
\bar{v}(z)
:=\sup_{ \eta \in \mathcal{S}} \mathbb{E}_{z}\bigg[    \int_0^{\eta}  e^{- (\rho+ m) s} \Big(Z_s h+m\,u(q)-m\,u(l\bar{B})\Big)ds  \bigg],
\end{align}
where $\bar{B}:=B+q$. Let
\begin{align}\label{disb}
\bar{b}= \frac{[m\, u(l\bar{B})-m\,u(q)]r \alpha_1}{(\rho+m)h(\alpha_1-1)} \quad \text{and} \quad \bar{C}_1=-\frac{h}{r \alpha_1}\bar{b}^{1-\alpha_1},
\end{align}
 and
\begin{equation}\label{disvz}
\bar{w}(z)= \left\{
\begin{aligned}
-\frac{h}{r \alpha_1}\Big( \frac{[m\, u(l\bar{B})-mu(q) ]r \alpha_1}{(\rho+m)h(\alpha_1-1)}\Big)^{1-\alpha_1} z^{\alpha_1}+ \frac{h}{r}z-\frac{m\,u(l\bar{B})-m\,u(q)}{\rho+m}, \quad \text{if} \ & z>\bar{b},\\
0, \quad \text{if} \ & 0<z\leq \bar{b}.
\end{aligned}
\right.
\end{equation}
\begin{theorem}\label{disverification}
The function $\bar{w}$ given by (\ref{disvz}) identifies with the dual value function $\bar{v}$ in (\ref{disos}). Moreover, the optimal stopping time takes the form 
\begin{align*}
	\eta^*=\eta(z;\bar{b}):=\inf \{t \geq 0: Z_t^z \leq \bar{b}\},
\end{align*}
with $\bar{b}$ as in (\ref{disb}). 
\end{theorem}

{In particular, we have the following observations:
\begin{itemize}
	\item When $0<\gamma<1$, the optimal time for life insurance purchase occurs when the agent's wealth surpasses a critical threshold. Also, when $q \to 0$, we have $\bar{b} \to b$, where $b$ is defined in (\ref{b}).
	\item When $\gamma>1$, there also exists an early exercise boundary $\bar{b}$, which is different from the result of Theorem \ref{buy}. However, when $q \to 0$, we have $\bar{b} \to +\infty$, which means that the agent should buy life insurance immediately (i.e.\ $\eta^*=0$). 
	
\end{itemize}
\subsubsection{A controlled bequest amount}\label{subsection5.2}
Similar to (\ref{vxy}), and from (\ref{eq5-2}) the agent aims at determining
\begin{align}\label{dis6-1}
\bar{V}^B(x',y):=\sup_{(c, \pi,B,\eta) \in \mathcal{A}^B(x',y) }\mathbb{E}_{x',y}\bigg[\int^\eta_0e^{-(\rho+m) t}\Big(u(c_t)+m\,u(q)\Big)dt \nonumber\\+\int^\infty_\eta e^{-(\rho+m) t} \Big(u(c_t)+m\,u(l(q+B_\eta))\Big) dt  \bigg].
\end{align}
In this section, we will focus on (\ref{dis6-1}), first, we conduct successive transformations (similar to Subsections \ref{sec4-2-1}, \ref{sec4-2-2} and \ref{sec4-2-3}) that connect the original stochastic control-stopping problem (with value function $\bar{V}^B$) with its dual problem (with value function $\bar{v}^B$). Then we study the reduced-version dual stopping problem (with value function $\widetilde{v}^B$).

	\textbf{The static budget constraint}. Recall that $x'= \phi x$. We replace $x$ by $x'$ in (\ref{6-2}) and (\ref{6-3}) and obtain the following static budget constraint:
	\begin{align}\label{dis6-2}
		\mathbb{E}_{x',y}\big[{\xi_{t\wedge \eta}}X_{t\wedge \eta}\big]+\mathbb{E}_{x',y}\bigg[ \int_0^{t\wedge \eta} \xi_{u} c_u du                         \bigg]  \leq x'+ \mathbb{E}_{x',y}\bigg[ \int_0^{t\wedge \eta} \xi_{u} Y_u du \Bigg], \quad \text{if} \ t \geq 0,
	\end{align}
	and
	\begin{align}\label{dis6-3}
		\mathbb{E}_{x',y}\big[{\xi_{t \vee \eta}}X_{t \vee \eta}\big]+\mathbb{E}_{x',y}\bigg[ \int_{\eta}^{t\vee \eta} \xi_u \, \Big(c_u+mB_0\Big) \, d u   \bigg]   \leq X_\eta+ \mathbb{E}_{x',y}\bigg[ \int_\eta^{t \vee \eta} \xi_{u} Y_u du \Bigg], \ \text{if} \ t\geq 0,
	\end{align}
	where $X$ is given by (\ref{stateeq}) with initial datum $x'$.
	
	\textbf{The optimization problem after purchasing life insurance.}
	Now, we look into the agent's optimization problem after the purchase of life insurance. Apart from determining consumption and portfolio choices, the agent is also required to determine the optimal bequest amount $B^*_{\eta}$ at purchasing time $\eta$, where $\eta=0$.
	 \\
	Letting $\mathcal{A}^B_0(x',y):= \{ (c,\pi,B): (c,\pi,B,0) \in \mathcal{A}^B(x',y)\},$ where the subscript $0$ indicates that the purchasing time $\eta$ is equal to $0$, the agent's value function after purchasing life insurance reads as 
	\begin{align*}
		\widetilde{V}^B(x',y):= \sup_{(c,\pi,B)\in \mathcal{A}_0^B(x',y)} \mathbb{E}_{x',y	}\bigg[\int^\infty_0 e^{-(\rho+m) s } \Big(u(c_s)+m \, u(l(q+B_0))\Big) ds \bigg].
	\end{align*}
	Recalling that $\xi_{s}=e^{-rs-\theta W_s -\frac{1}{2}\theta^2 s}$, for any pair $(c,\pi,B)\in \mathcal{A}_0^B(x',y)$ with a Lagrangian multiplier $z>0$, we have
	\begin{align*}
		\mathbb{E}&_{x',y} \bigg[    \int_0^{\infty}  e^{-(\rho+m) s } \Big(u(c_s)+m  \, u(l(q+B_0))\Big)  ds  \bigg] \nonumber \\ \leq &\
		\mathbb{E}_{x',y} \bigg[    \int_0^{\infty} e^{-(\rho+m) s } \Big(u(c_s)+m  \, u(l(q+B_0))\Big)  ds  \bigg]   \nonumber\\ &- z      \mathbb{E}_{x',y}\bigg[  \int_0^{\infty} \xi_{s} \Big(c_s+mB_0-Y_s\Big) ds                         \bigg]  +z x'                  \nonumber               \\
		\leq &\ \mathbb{E}_{x',y}\bigg[    \int_0^{\infty } e^{-(\rho+m) s }\Big(\widehat{u}(zP_s)+ zP_sY_s\Big) ds  \bigg]+ z x'+\widetilde{u}(z),
	\end{align*}
	where the first inequality results from the budget constraint stated in \eqref{dis6-3}, $P$ and $\widehat{u}$ are defined in (\ref{pt}), and
	\begin{align*}
		\quad \widetilde{u}(z)&:= \sup_{B_0 \geq 0}\mathbb{E}_{x',y}\bigg[   \int_0^{\infty }  e^{-(\rho+m) s } \Big(m \, u(l(q+B_0))-z P_s m B_0\Big) ds\bigg] \nonumber \\&
		= \sup_{B_0 \geq 0} \Big[\frac{m \, u(l(q+B_0))}{\rho+m}-z \frac{mB_0}{r}\Big].
	\end{align*}
	Let $\overline{L}:=\frac{q^{-\gamma} l^{1-\gamma}r}{\rho+m}$, and define the bequest amount solving the previous maximization is 
	\begin{equation*}
		B_0^*:=\left\{
		\begin{aligned}
			\Big[\frac{z(\rho+m)}{r}\Big]^{-\frac{1}{\gamma}}l^{\frac{1-\gamma}{\gamma}} -q, \quad &\text{if} \ z< \overline{L},
			\\
			0, \quad &\text{if} \ z\geq \overline{L},
		\end{aligned}
		\right. 
	\end{equation*}
	and 
	\begin{equation*}
		\widetilde{u}(z):= \left\{
		\begin{aligned}
			m l^{\frac{1-\gamma}{\gamma}} \frac{\gamma}{1-\gamma}r^{\frac{1-\gamma}{\gamma}}(\rho+m)^{-\frac{1}{\gamma}}z^{\frac{\gamma-1}{\gamma}}+\frac{zmq}{r} ,\quad \text{if} \ z< \overline{L},	\\
			\frac{m u(lq)}{\rho+m},\quad \text{if} \ z\geq \overline{L}.
		\end{aligned}
		\right.
	\end{equation*}
	In particular, it is easy to check that $\widetilde{u}'(z)<0$ when $z<\overline{L}$ and $\widetilde{u}'(\overline{L})=0$. Also, $\widetilde{u}''(z)>0$ when $z<\overline{L}$.
	Since the amount of bequest will be fixed once purchased, the (candidate) optimal bequest amount $B^*_0$ will just depend on the initial state $z$. This means the agent can choose the optimal bequest amount based on their initial state $z$ at the time of purchase ($\eta=0$). 
	Next we set
	\begin{align*}
		\widetilde{Q}^B(z,y):=  \mathbb{E}_{z,y}\bigg[    \int_0^{\infty }  e^{- (\rho+ m) s} \Big[\widehat{u}(Z_s)+Z_sY_s\Big] ds  \bigg],
	\end{align*}
	where $Z$ is given in (\ref{Z}).
	\begin{pro}
		One has $\widetilde{Q}^B \in  C^{2,2}(\mathbb{R}^2_+).$ Moreover, $\widetilde{Q}^B$  satisfies 
		\begin{align}\label{dis6-9}
			-{\mathcal{L}}\widetilde{Q}^B=\widehat{u}+zy, \ \text{on} \ \mathbb{R}^2_+, 
		\end{align}
		where $\mathcal{L}$ is defined in (\ref{L}).
	\end{pro}
	\begin{proof}
		The proof is similar to the proof of Proposition \ref{pro4.1}, we thus omit it. 
	\end{proof}
	Similar to Theorem \ref{thedual}, we also have the following duality relations. 
	\begin{theorem}\label{disthedual}
		The following dual relations hold:
		\begin{align*}
			\widetilde{V}^B(x',y) = \inf_{z >0} [ \widetilde{Q}^B(z,y)+zx'+\widetilde{u}(z) ], \quad \widetilde{Q}^B(z,y) =  \sup_{x'> -\frac{y}{\kappa}+\frac{m B_0}{r}}[\widetilde{V}^B(x',y)-zx'-\widetilde{u}(z)].
		\end{align*}
	\end{theorem}
	
	\textbf{The dual optimal stopping problem}. Now, for any $(x',y) \in \mathcal{O}$ and Lagrange multiplier $z>0$, from (\ref{dis6-1}) and the budget constraint (\ref{dis6-2}), recalling that $P_s$ as in (\ref{pt}), we have by the strong Markov property
	\begin{align*}
		\mathbb{E}&_{x',y}\bigg[\int^\eta_0 e^{-(\rho+m)t }(u(c_t)+m\,u(q))dt +\int^\infty_\eta e^{-(\rho+m)t } \Big(u(c_t)+m\,u(l(B_\eta+q))\Big) dt  \bigg]
		\nonumber \\ \leq &\
		\sup_{(c,\pi,B,\eta) \in \mathcal{A}^B(x',y)} \mathbb{E}_{x',y}  \bigg[\int^\eta_0 e^{-(\rho+m) t }(u(c_t)+m\,u(q)) dt + e^{-(\rho+m) \eta }\widehat{V}^B(X_\eta,Y_\eta)  \bigg] \nonumber \\& - z      \mathbb{E}_{x',y} \bigg[   \xi_{\eta}X_{\eta}+ \int_0^{\eta} \xi_{t}\Big(c_t-Y_t\Big) dt                         \bigg]  +z x'                        \nonumber       \\
		\leq &\  \sup_{\eta \in \mathcal{S}} \mathbb{E}_{z,y}\bigg[    \int_0^{\eta} e^{-(\rho+m) s }\Big(\widehat{u}(Z_s)+Z_sY_s+m\,u(q)\Big) ds
		+  e^{-(\rho+m) \eta }\Big(\widehat{Q}^B(Z_{\eta},Y_\eta)+\widetilde{u}(Z_\eta)\Big) \bigg]+ zx' , 
	\end{align*}
	where $\widehat{u}(z)= \sup_{c \geq 0}[u(c)-c z]$ and $Z_t $ is defined in (\ref{Z}). 
	Hence, setting
	\begin{align}\label{dis6-17}
		\bar{v}^B(z,y):=&\ \sup_{ \eta \in \mathcal{S}
		} \mathbb{E}_{z,y}\bigg[    \int_0^{\eta }  e^{- (\rho+ m) s}\Big(\widehat{u}(Z_s) +Z_sY_s+m\,u(q)\Big)ds \nonumber\\&+ e^{- (\rho+ m) \eta}\Big(\widehat{Q}^B(Z_{\eta},Y_\eta)+\widetilde{u}(Z_\eta)\Big) \bigg],
	\end{align}
	we have a two-dimensional optimal stopping problem, with dynamic $(Z,Y)$ as in (\ref{Z}) and (\ref{income}). \\
	In the following, we study (\ref{dis6-17}). Before doing that, similarly to Theorem \ref{dualrelation2}, we have the following theorem that establishes a dual relation between the original problem (\ref{dis6-1}) and the optimal stopping problem (\ref{dis6-17}).
	\begin{theorem}
		The following duality relations hold:
		\begin{align*}
			\bar{V}^B(x',y) = \inf_{z>0}[\bar{v}^B(z,y)+zx'], \quad \bar{v}^B(z,y) = \sup_{x'>-\frac{y}{\kappa}} [\bar{V}^B(x',y)-zx'].
		\end{align*}  
	\end{theorem}
	To study the optimal stopping problem (\ref{dis6-17}), it is convenient to introduce the function
	\begin{align}\label{dis6-18}
		\widetilde{v}^B(z,y):=\bar{v}^B(z,y)-{\widetilde{Q}^B}(z,y)-\widetilde{u}(z).
	\end{align}
	Applying It\^o's formula to $\{ e^{-(\rho+m)t }[\widetilde{Q}^B(Z_t,Y_t)+\widetilde{u}(Z_t)] , t \in [0,\eta]  \}$, and taking conditional expectations we have 
	\begin{align*}
		\mathbb{E}&_{z,y}\Big[ e^{-(\rho+m) \eta }\Big(\widetilde{Q}^B(Z_{\eta},Y_\eta)+\widetilde{u}(Z_\eta)\Big) \Big]= \widetilde{Q}^B(z,y)+\widetilde{u}(z)\\&+ \mathbb{E}_{z,y}\bigg[ \int_0^{\eta}  e^{-(\rho+m) s }\mathcal{L}\Big(\widetilde{Q}^B( Z_s,Y_s)+\widetilde{u}(Z_s)\Big) ds               \bigg ],
	\end{align*}
	where $\mathcal{L}$ is defined in (\ref{L}). Combining (\ref{dis6-17}) and (\ref{dis6-18}), we have 
	\begin{align*}
		\widetilde{v}^B(z,y) =&\ \sup_{  \eta \in \mathcal{S}
		} \mathbb{E}_{z,y}\bigg[    \int_0^{\eta }  e^{- (\rho+ m) s} \Big(\widehat{u}(Z_s)+m\,u(q)+Z_sY_s\Big) ds  \nonumber \\
		&+ \int_{0}^{\eta} e^{- (\rho+ m) s}\mathcal{L}\Big(\widetilde{Q}^B(Z_s,Y_s)+\widetilde{u}(Z_s)\Big) ds               \bigg ] \nonumber \\
		=&\ \sup_{ \eta \in \mathcal{S}} \mathbb{E}_{z,y}\bigg[    \int_0^{\eta}  e^{- (\rho+ m) s} \Big(m\,u(q)-K \widetilde{u}(Z_s)+\frac{Z_smqK}{r}\Big)\mathds{1}_{\{Z_s < \bar{L} \}} ds \nonumber \\ &+  \int_0^{\eta}  e^{- (\rho+ m) s} \Big(m\,u(q)-m\,u(lq)\Big)\mathds{1}_{\{Z_s \geq \bar{L} \}} ds\bigg],
	\end{align*}
	where $K= \frac{1}{\gamma}( \rho+m-r(1-\gamma)-\frac{1-\gamma}{2\gamma} \theta^2)>0$ due to Assumption \ref{assume2}, and where we have used the fact that (cf.\ (\ref{dis6-9}))
	\begin{align*}
		\mathcal{L}(\widetilde{Q}^B(z,y)+\widetilde{u}(z))&=-\widehat{u}(z)-zy-K \widetilde{u}(z)+ \frac{zKmq}{r}, \quad z< \overline{L},
	\end{align*}
	and 
	\begin{align*}
		\mathcal{L}(\widetilde{Q}^B(z,y)+\widetilde{u}(z))&=-\widehat{u}(z)-zy-mu(lq), \quad z\geq  \overline{L}.
	\end{align*}
	By the dynamic programming principle, we expect that $\widetilde{v}^B$ identifies with a suitable solution $\widetilde{w}^B$ to the HJB equation
	\begin{align}\label{dishjb1}
		\max \{ [\mathbb{L}_Z-(\rho+m)]\widetilde{w}^B+(m\,u(q)-K\widetilde{u}(z)+\frac{zKmq}{r})\mathds{1}_{\{z<\overline{L}\}}\nonumber\\+ (m\,u(q)-m\,u(lq))\mathds{1}_{\{z\geq \overline{L}\}}, -\widetilde{w}^B  \} =0, \quad z>0,
	\end{align}
	where $\mathbb{L}_Z$ is given by (\ref{operator}).}

As usual in optimal stopping theory, we let
\begin{align*}
\mathcal{\widetilde{ C}}:=\{ z \in \mathbb{R}_+: \widetilde{v}^B(z)>0  \},\quad
\mathcal{\widetilde{ R}}:=\{ z \in \mathbb{R}_+: \widetilde{v}^B(z)=0  \}
\end{align*}
be the so-called continuation (waiting) and stopping (purchasing) regions, respectively. 
	


We expect that the optimal purchasing strategy should be of barrier-type and we guess that $\widetilde{w}$ satisfies 
\begin{equation}\label{disguess}
	\left \{
	\begin{aligned}
		\widetilde{w}(z)=0, \quad \forall &z\in(0,\widetilde{b}],\\
		\Big(\mathbb{L}_Z-(\rho+m)\Big)\widetilde{w}(z)+mu(q)-K \widetilde{u}(z)+\frac{z K mq}{r} =0,\quad \forall &z\in (\widetilde{b},\bar{L}], \\
\Big(\mathbb{L}_Z-(\rho+m)\Big)\widetilde{w}(z)+mu(q)-mu(lq) =0,\quad \forall &z\in (\bar{L},+\infty), \\\end{aligned}
	\right.
\end{equation}
Solving (\ref{disguess}), we have 
\begin{equation}\label{dissolution}
	\widetilde{w}(z)=\left \{
	\begin{aligned}
		0, \quad \forall &z\in(0,\widetilde{b}],\\
		A_1 z^{\alpha_1}+A_2 z^{\alpha_2}+\frac{mu(q)}{\rho+m} -Cz^{\frac{\gamma-1}{\gamma}},\quad \forall &z\in (\widetilde{b},\bar{L}], \\
		B_1 z^{\alpha_1} +\Delta, \quad \forall &z\in (\bar{L},+\infty),			\end{aligned}
	\right.
\end{equation}
where $\Delta:= \frac{m u(q)-mu(lq)}{\rho+m}$ and $C:=ml^{\frac{1-\gamma}{\gamma}} \frac{\gamma}{1-\gamma}r^{\frac{1-\gamma}{\gamma}}(\rho+m)^{-\frac{1}{\gamma}}$, $A_1,A_2$ and $B_1$ are undetermined constants and  $\alpha_1<0< 1<\alpha_2$ are the real roots of the algebraic equation
\begin{align*}
\frac{1}{2}\theta^2 \alpha^2+(\rho-r+m-\frac{1}{2}\theta^2)\alpha-(\rho+m)=0.
\end{align*}
Then we appeal to the so-called ``smooth fit principle", which dictate that the candidate value function $\widetilde{w}(z)$ should be $C^1$ in $\widetilde{b}$ and $\bar{L}$. These conditions give rise to the system of equations
\begin{equation}\label{eq5-15-2}
	\begin{cases}
		A_1{\widetilde{b}}^{\alpha_1}+A_2 {\widetilde{b}}^{\alpha_2}+ \frac{m u(q)}{\rho+m}-C {\widetilde{b}}^{\frac{\gamma-1}{\gamma}}=0,\\
		A_1 \alpha_1 {\widetilde{b}}^{\alpha_1-1}+A_2 \alpha_2 {\widetilde{b}}^{\alpha_2-1}-C \frac{\gamma-1}{\gamma}
		{\widetilde{b}}^{\frac{-1}{\gamma}}=0,\\
		A_1 \bar{L}^{\alpha_1}+A_2 \bar{L}^{\alpha_2}+ \frac{m u(q)}{\rho+m}-C \bar{L}^{\frac{\gamma-1}{\gamma}}=B_1 \bar{L}^{\alpha_1}+\Delta,\\
		A_1 \alpha_1 \bar{L}^{\alpha_1-1}+A_2 \alpha_2 \bar{L}^{\alpha_2-1}-C \frac{\gamma-1}{\gamma} \bar{L}^{\frac{-1}{\gamma}}=B_1 \alpha_1 \bar{L}^{\alpha_1-1}.		 	\end{cases}
\end{equation}

\begin{theorem}\label{theorem5-5}
	Suppose that there exists a solution $\widetilde{b}<\bar{L}$ solving the system of equations (\ref{eq5-15-2}) and $m u(\widetilde{b})-K \widetilde{u}(\widetilde{b})+\frac{\widetilde{b}Kmq}{r}\leq 0 $, $\widetilde{w}(z)\geq 0$ for any $z$.  The function $\widetilde{w}$ given by (\ref{dissolution}) identifies with the dual value function $\widetilde{v}^B$ and the optimal purchasing time is $\eta^*(z;\widetilde{b})=\inf\{t\geq 0: Z^z_t \leq \widetilde{b} \}$.
	\end{theorem}
\begin{proof}
	The proof is given in Appendix \ref{prooftheorem5-5}.
\end{proof}

To ensure the verifiability of the conditions in Theorem \ref{theorem5-5}, Figure \ref{fig9} illustrates the value function as defined in (\ref{dissolution}). It is evident from the figure that $\widetilde{b}<\bar{L}$. Moreover, our numerical example suggests that $\widetilde{b}<\bar{L}$ holds over a broader range of parameters, which further supports the validity of Theorem \ref{theorem5-5}.

\begin{figure}[htbp]
		\setlength{\abovecaptionskip}{0pt}
	\setlength{\belowcaptionskip}{5pt}
		\centering
		\includegraphics[width=0.5\textwidth]{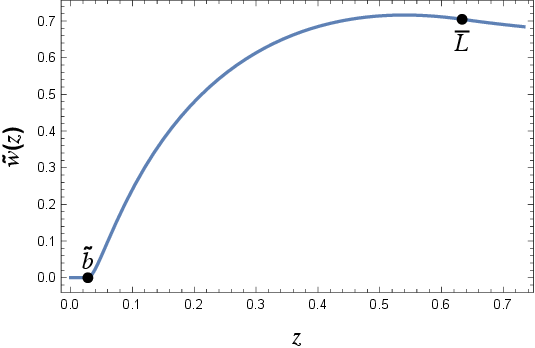}
		\caption{Value function in (\ref{dissolution}) with parameters: $\mu=0.05, r=\rho=0.01, \sigma=0.22, m=0.0175,\gamma=1.8,q=1,l=0.5$. }
		\label{fig9}
\end{figure}
However, for certain parameter choices, numerical solutions of the system of equations (\ref{eq5-15-2}) suggest that $\widetilde{b}\geq \bar{L}$. In these cases, we expect that the optimal purchasing rule may be triggered by two or more boundaries. A rigorous analysis of this possibility is beyond the scope of the current work and is left for future research.


\subsection{Age-dependent force of mortality}

In this section, we extend the previous constant force of mortality case to the age-dependent force of mortality. For simplicity, we consider the bequest amount to be predetermined, i.e., $B$ is a given constant. Now we assume an agent whose force of mortality rate evolves according to the standard Gompertz model (see, e.g., \cite{gompertz1825nature}). The process $M:=\{M_s, s\geq t\}$ thus follows the dynamics
 \begin{align}\label{2-1}
d M_s = aM_s ds ,  \ M_t =m>0,
\end{align}
with $a>0$. As in section \ref{sec5-1}, we assume the policyholder divides her initial wealth into two parts: $x=(1-\phi)x+ \phi x:= q + x^{\prime}$, where $q:=(1-\phi)x>0 $ is directly designated for inheritance, while $ x^{\prime}:=\phi x$ is allocated for the policyholder's personal consumption and life insurance decisions. The agent's wealth $X$ evolves as 
\begin{align*}
	d X^{c,\pi, \eta}_t =  [\pi_t(\mu-r)+rX^{c,\pi,\eta}_t-c_t-h_t \mathds{1}_{\{t \geq \eta \}}+Y_t]dt +\pi_t \sigma dW_t, \quad X_0^{c,\pi, \eta}= x',\end{align*}
where we assume $h_t={ M_t}\cdot B$ for simplicity.  The agent's aim is then to maximize the following expected lifetime utility 
\begin{align}
    \mathbb{E}\bigg[ \Big( \int^\tau_0e^{-\rho t}u(c_t)dt +e^{-\rho \tau}u(q) \Big)\mathds{1}_{\{ \eta \geq \tau\}} + \mathds{1}_{\{\eta<\tau\}}\Big( \int^\tau_0 e^{-\rho t}u(c_t)dt+e^{-\rho \tau}u(l(q+B)) \Big) \bigg],
\end{align}
where $q>0$ is a given constant. Then we introduce the set of admissible strategies $\mathcal{A}^M(x,y,m)$: 
\begin{definition}
Let $(x,y,m) \in \mathcal{O} \times \mathbb{R}_+ $ be given and fixed. The triplet of choices $(c,\pi, \eta)$ is called an \textbf{admissible strategy} for $(x,y,m)$, and we write $(c,\pi, \eta) \in \mathcal{A}^M(x,y,m)$, if it satisfies the following conditions: 
\begin{enumerate}[label=(\roman*)]
	\item  $c$ and $\pi$ are progressively measurable with respect to $\mathbb{F}$, $\eta \in \mathcal{S}$;
	\item $c_t \geq 0$ for all $t \geq 0$ and $\int_0^t(c_s+|\pi_s|^2)ds <\infty$ for all $t\geq 0$ $ \ \mathbb{P}$-a.s.; 
	\item $X^{c,\pi, \eta}_t +g_t >  \frac{B \cdot M_t}{r}\mathds{1}_{\{t\geq  \eta\}}$ for all $t\geq 0$, where $g_t$ is defined in (\ref{gt}).
\end{enumerate}
\end{definition}

Similar to (\ref{eq5-2}), given the Markovian setting, the agent aims at determining
\begin{align}\label{eq5-15-1}
V^M(x,y,m):=&\sup_{(c, \pi,\eta) \in \mathcal{A}^M }\mathbb{E}_{x,y,m}\bigg[\int^\eta_0e^{-\int^t_0(\rho+M_s) ds}\Big(u(c_t)+M_t \,u(q)\Big)dt \nonumber\\  &+ \int^\infty_\eta e^{-\int^t_0(\rho+M_s) ds} \Big(u(c_t)+M_t\,u(l(q+B))\Big) dt  \bigg].
\end{align}

In the following, we shall focus on (\ref{eq5-15-1}). Similar to our previous procedure, we could also use the duality method to transform the original control-stopping problem (\ref{eq5-15-1}) to the pure stopping problem and establish the corresponding duality relations. Now the dual value function is given by 
\begin{align*}
	J^M(z,m):=\sup_{\eta}\mathbb{E}_{z,m}\bigg[\int^\eta_0 e^{-\int^t_0 (\rho+M_s)ds}\Big(M_t u(q)-M_t u(l(q+B))+Z_t{ M_t} B\Big)dt \bigg],
\end{align*}
where $Z_t$ satisfies 
\begin{align*}
	dZ_t=  (\rho-r+ M_t)Z_tdt - \theta Z_t dW_t, \quad Z_0= z.
	\end{align*}
When $M_t\equiv m$, the above dual optimal stopping problem reduces to (\ref{disos}). However, as the problem is now two-dimensional $(Z,M)$, the classical ``guess and verify" method is inapplicable. Following \cite{ferrari2023optimal}, we use a probabilistic approach with free-boundary analysis.

As usual in optimal stopping theory, we let
\begin{align*}
\mathcal{C}^M:=\{ (z,m) \in \mathbb{R}_+^2: J^M(z,m)>0  \},\quad
\mathcal{R}^M:=\{ (z,m) \in \mathbb{R}_+^2: J^M(z,m)=0  \}
\end{align*}
be the so-called continuation (waiting) and stopping (purchasing) regions, respectively. We denote by $\partial \mathcal{C}^M$ the boundary of the set $\mathcal{C}^M.$
Then we introduce the stopping time 
\begin{align*}
\eta^*:= \inf\{t \geq 0 :(Z_t,M_t) \in \mathcal{R}^M\}, \quad \mathbb{P}_{z,m}\text{-}a.s.,
\end{align*}
with $\inf \emptyset = + \infty$, one has that $\eta^*$ is optimal for $J^M(z,m)$ for any $(z,m)\in \mathbb{R}_+^2$.

Now we show that the boundary $\partial \mathcal{C}^M$ can be represented by a function $b^M(m)$. Since the monotonicity of $J^M$ with respect to $z$, we have the following result:
\begin{pro}\label{ageboundary}
	There exists a function $b^M:\mathbb{R}_+ \mapsto (0,\infty)$ such that 
	\begin{align*}
		\mathcal{R}^M=\{(z,m)\in \mathbb{R}^2_+: 0<z\leq b^M(m)\}.
	\end{align*}
\end{pro}

Different from the constant force of mortality case in which the boundary is a constant, now the boundary is a function of the force of mortality process. Following the analysis of \cite{ferrari2023optimal}, we could also determine a nonlinear integral equation that uniquely characterizes the free boundary.  
\begin{pro}
	The optimal boundary $b^M$ is the unique continuous solution to the following nonlinear integral equation: For all $m\in \mathbb{R}_+$,
	\begin{align}\label{eq5-19}
		0= \mathbb{E}_{b^M(m),m}\bigg[\int^\infty_0 e^{-\int^t_0 (\rho+M_s)ds}\Big(M_t u(q)-M_t u(l(q+B))\nonumber\\+Z_tM_t B\Big)\mathds{1}_{\{ Z_t\geq b^M(M_t)  \} } dt \bigg], \ m>0.
			\end{align}
\end{pro}


Summarizing these findings, we now explore an alternative scenario where the policyholder sets aside a portion of her initial wealth for inheritance. Within this framework, we consistently identify a wealth threshold beyond which purchasing life insurance becomes optimal, regardless of whether the force of mortality is constant or age-dependent (cf.\ Theorems \ref{disverification}, \ref{theorem5-5}, and Proposition \ref{ageboundary}). This result not only strengthens our previous conclusions but also provides new economic insights into the optimal timing for acquiring life insurance.
}

{It is worth noting that in this section we consider only the predetermined-bequest case and formally derive its corresponding integral equation for free boundary. The controlled-bequest case---with an age-dependent force of mortality---poses substantially greater analytical challenges and requires rigorous mathematical treatment; we therefore defer it to future work.}

\section{Conclusions} \label{sec:con}
{We have investigated}  the optimal timing of life insurance purchase for an agent facing uncertain lifetime and stochastic labor income. The agent can make a choice regarding when to buy life insurance, considering two types of bequests: One with a predetermined amount and the other granting the agent the freedom to determine the bequest amount as an additional variable. The optimization problem is formulated as a stochastic control-stopping problem over a random time horizon, which contains two state variables: Wealth and labor income. We have solved both cases using dual transformation and free-boundary approach, and obtained the analytical solutions for the value functions and optimal policies. We find there are different optimal life insurance purchasing strategies in the two cases and the risk aversion parameter $\gamma$ plays a crucial role. For example, when given a predetermined bequest amount, the agent should buy life insurance whenever her wealth exceeds a labor income-dependent optimal stopping boundary if $\gamma<1$, whereas life insurance should be bought immediately when $\gamma>1$. A detailed numerical study allows to draw interesting economic implications about the sensitivity of the optimal purchasing boundary and the optimal bequest amount with respect to the model's parameters.

{We suggest} several avenues for potential extensions. 
{For example, considering the effects of borrowing constraints (see, e.g., \cite{zeng2016optimal}) or non-hedgeable labor income risk on optimal life insurance decisions could yield valuable insights.} Moreover, incorporating the impact of health shocks in an individual's optimization problem presents an interesting yet complex challenge due to life's unpredictability and the potential for changing health status over time.

%
%
%
%

%

\appendix
\setcounter{subsection}{0}
\renewcommand\thesubsection{A.\arabic{subsection}}
\setcounter{equation}{0}
\renewcommand\theequation{A.\arabic{equation}}
\setcounter{lemma}{0}
   \renewcommand{\thelemma}{\Alph{section}.\arabic{lemma}}

\section{Technical proofs and auxiliary results}\label{appendix}
\renewcommand\theequation{A.\arabic{equation}}
\renewcommand\thesubsection{A.\arabic{subsection}}
\counterwithin{equation}{section}
\counterwithin{subsection}{section}

\subsection{Proof of Proposition \ref{pro3.1} }\label{proofpro3.1}
\begin{proof}
	First, we compute the convex dual of $u$ from (\ref{utility}) (cf.\ (\ref{pt})); that is,
\begin{align}\label{uhat}
\widehat{u}(z)=\frac{\gamma}{1-\gamma}z^{\frac{\gamma-1}{\gamma}}, \quad z>0.
\end{align}

{\raggedleft{Therefore,}} by (\ref{W}) and (\ref{uhat}) we rewrite $\widehat{Q}(z,y)$ as follows
\begin{align}\label{W1}
\widehat{Q}(z,y)=&\  \mathbb{E}_{z,y}\bigg[    \int_0^{\infty }  e^{-(\rho+m)s} \frac{\gamma}{1-\gamma}{Z_s}^{\frac{\gamma-1}{\gamma}}  ds  \bigg]+\mathbb{E}_{z,y}\bigg[ \int^\infty_0 e^{-(\rho+m)s} Z_sY_sds  \bigg]\nonumber\\
&+   \mathbb{E}\bigg[    \int_0^{\infty }  e^{- (\rho+ m) s} m \,  u(lB) ds  \bigg]  \nonumber\\
=&\  z^{\frac{\gamma-1}{\gamma}} \frac{\gamma}{1-\gamma}\int_0^{\infty }  e^{-(\rho+m)s}e^{[\frac{\gamma-1}{\gamma} (\rho+m-r-\frac{1}{2}\theta^2)+ \frac{1}{2}\frac{(\gamma-1)^2 \theta^2}{\gamma^2}] s} ds+\frac{zy}{\kappa} +\frac{m\,u(lB)}{\rho+m}, 
\end{align}
where we have used the explicit expression of $Z$. 
Moreover, due to Assumption \ref{assume2} and (\ref{W1}), we can verify that
\begin{align}\label{W2}
\widehat{Q}(z,y)=  z^{\frac{\gamma-1}{\gamma}} \frac{\gamma}{1-\gamma} \frac{1}{K}+\frac{zy}{\kappa} +\frac{m\,u(lB)}{\rho+m} < \infty,
\end{align} 
where $K:= \frac{1}{\gamma}( \rho+m-r(1-\gamma)-\frac{1-\gamma}{2\gamma} \theta^2)>0$ due to Assumption \ref{assume2}.  
Finally, it is easy to see that $\widehat{Q} \in  C^{2,2}(\mathbb{R}^2_+)$ from (\ref{W2}) and it satisfies (\ref{LW}) by direct calculations.
\end{proof}

{
\subsection{Proof of Theorem \ref{dualrelation1}}\label{proofdualrelation1}
\begin{proof}
	First, we show the duality relations. Since $(c,\pi) \in \mathcal{A}_0(x,y)$ is arbitrary, taking the supremum over $(c,\pi)  \in \mathcal{A}_0(x,y)$ on the left-hand side of (\ref{dual}) and recalling (\ref{vhat}), we get, for any $z>0$,
\begin{align*}
\widehat{V}(x,y) \leq \widehat{Q}(z,y)+z(x-\frac{h}{r}),  
\end{align*}
and thus 
\begin{align}\label{eqa-4-1}
\widehat{Q}(z,y) \geq \sup_{x> \frac{h}{r}-\frac{y}{\kappa}}[\widehat{V}(x,y)-z(x-\frac{h}{r})] \quad \text{and} \quad  \widehat{V}(x,y) \leq \inf_{z>0}[\widehat{Q}(z,y)+z(x-\frac{h}{r})].
\end{align}
Let $y\in \mathbb{R}_+$ be given. Now,  we define 
\begin{align*}
	\mathcal{X}(z):= \mathbb{E}_{z,y}\bigg[  \int_0^{\infty} \xi_{s} \Big(\mathcal{I}^u(zP_s)-Y_s\Big) ds                         \bigg].
	\end{align*}
Due to $\mathcal{I}^u(x)=x^{-\frac{1}{\gamma}}$ and monotone convergence theorem, it is easy to show that the function $\mathcal{X}+\frac{h}{r}$ is strictly decreasing and continuous on $(0,\infty)$, with $\lim_{z\downarrow 0}(\mathcal{X}(z)+\frac{h}{r})=\infty$ and  $\lim_{z\to +\infty}(\mathcal{X}(z)+\frac{h}{r})=\frac{h}{r} -\frac{y}{\kappa}$. Thus, for given $x>\frac{h}{r}-\frac{y}{\kappa}$ there exists a unique $z^*>0$ such that $x=\mathcal{X}(z^*)+\frac{h}{r}$.


{\raggedleft{For}} the reverse inequalities, observe that the equality in (\ref{dual}) holds if and only if 
\begin{align}\label{3-10}
c_s= \mathcal{I}^u(zP_s),
\end{align}
and 
\begin{align}\label{3-11}
 \mathbb{E}_{x,y}\bigg[  \int_0^{\infty} \xi_{s}\Big(c_s-Y_s\Big) ds                         \bigg]  =  x -\frac{h}{r}.
\end{align}
In particular, the equality holds in (\ref{dual}) if (\ref{3-10})
 is satisfied and $x-\frac{h}{r}=\mathcal{X}(z^*)$. This gives $\widehat{Q}(z,y)= \widehat{V}(\mathcal{X}(z^*)+\frac{h}{r},y)-z \mathcal{X}(z^*) \leq \sup_{x>\frac{h}{r}-\frac{y}{\kappa}} \{\widehat{V}(x,y) -z(x-\frac{h}{r})\} $. Combining with (\ref{eqa-4-1}), this completes the proof of $\widehat{Q}(z,y) = \sup_{x>\frac{h}{r}-\frac{y}{\kappa}}\{\widehat{V}(x,y)-z(x-\frac{h}{r})\} $ and shows that, for $z>0$, the maximum in $\widehat{Q}(z,y) \leq  \sup_{x>\frac{h}{r}-\frac{y}{\kappa}}[\widehat{V}(x,y)-z(x-\frac{h}{r})]$ is attained by $x=\mathcal{X}(z^*)+\frac{h}{r}$.
 
 For any $x\in (\frac{h}{r}-\frac{y}{\kappa},\infty)$, we introduce the candidate optimal consumption process $c^*_s:=\mathcal{I}^u(z^*P_s). $
 Then 
 \begin{align}\label{eqa-7-1}
 	\mathbb{E}_{x,y}\bigg[  \int_0^{\infty} \xi_{s}\Big(c_s^*-Y_s\Big) ds                         \bigg]  = \mathcal{X}( z^*)= x -\frac{h}{r}.
 \end{align}
Moreover, Lemma \ref{budget2} guarantees the existence of a candidate optimal portfolio process $\pi^*_s$ such that $(c^*,\pi^*) \in \mathcal{A}_0(x,y)$. By Theorem 3.6.3 in \cite{karatzas1998methods} or Lemma 6.2 in \cite{karatzas2000utility}, one can then show that $(c^*, \pi^*)$ is optimal for the  problem $\widehat{V}$.

Now we present the explicit forms of optimal strategies and the corresponding optimal wealth process. 
From (\ref{eqa-7-1}), (\ref{W2}) and the strong Markov property we have 
\begin{align*}
	X_t^*=\mathcal{X}(Z_t^*)+\frac{h}{r} =-\widehat{Q}_z(Z_t^*,Y_t)+\frac{h}{r}=\frac{(Z_t^*)^{-\frac{1}{\gamma}}}{K}-\frac{Y_t}{\kappa}+\frac{h}{r},
\end{align*}
where $Z_t^*$ is the solution to Equation (\ref{Z}) with the initial condition $z^*$ satisfying $\widehat{Q}_z(z^*,y)+x-\frac{h}{r}=0$.
We then apply It\^o's formula to the optimal wealth process $X^*_t$, 
\begin{align*}
	d X_t^*= \Big[-\frac{1}{\gamma K}Z_t^{-\frac{1}{\gamma}} (\rho-r+m)+\frac{1}{2}\frac{1}{\gamma K} (\frac{1}{\gamma}+1)\theta^2 Z_t^{-\frac{1}{\gamma}}-\frac{\mu_y Y_t}{\kappa} \Big]dt+ \Big[\frac{\theta}{\gamma K}Z_t^{-\frac{1}{\gamma}}-\frac{\sigma_y Y_t}{\kappa} \Big]dW_t.
\end{align*}
Comparing the above equation with the dynamics of wealth in (\ref{wealth}), we have 
\begin{align*}
	\pi^*_t= \frac{\theta}{\gamma \sigma K}(Z_t^*)^{-\frac{1}{\gamma}}-\frac{\sigma_y Y_t}{\kappa \sigma}.
	\end{align*}
	Since $\mathcal{I}^u(x)=x^{-\frac{1}{\gamma}}$, we have optimal consumption process $c_t^*=(Z_t^*)^{-\frac{1}{\gamma}}$.\end{proof}
}

\subsection{Proof of Proposition \ref{finite}}\label{prooffinite}
\begin{proof}
From (\ref{jzy}) and (\ref{4-3}), it is clear that $\widehat{v}(z) \geq 0$ for all $z \in \mathbb{R}_+$. Moreover, from (\ref{hatjz}), we find that
\begin{align*}
&  \sup_{ \eta \in \mathcal{S} } \mathbb{E}_{z}\bigg[\int^\eta_0   Z_t h dt  -\int^\eta_0 e^{-(\rho+m)t} m\, u(lB)dt  \bigg] \nonumber \\
&\leq \mathbb{E}_{z}\bigg[ \int_0^{\infty }  z   h e^{-rt-\frac{1}{2}\theta^2 t-\theta W_t}dt  \bigg] = z\bigg[ \int_0^{\infty } he^{-rt}dt  \bigg] \leq \frac{zh}{r},
\end{align*} 
 which implies the claim.  
\end{proof}

\subsection{Proof of Theorem \ref{verification}}\label{proofverification}
\begin{proof}
The proof is organized in two steps.

\vspace{6pt}
\textbf{Step 1:} First we show that $\widehat{w}(z)$ in (\ref{vz}) satisfies the HJB equation (\ref{hjb1}). By construction, we only need to show that 
\begin{align}\label{4-15}
	\Big(\mathbb{L}_Z-(\rho+m)\Big)\widehat{w}(z)+hz-m\,u(lB)\leq 0, \quad \forall z\leq b,
	\end{align}
and
\begin{align}\label{4-16}
	\widehat{w}(z)\geq 0, \quad \forall z> b.
	\end{align}
	To prove (\ref{4-15}), we define $F(z):=(\mathbb{L}_Z-(\rho+m))\widehat{w}(z)+hz-m\ u(lB)$, which is such that $F(z)=hz-m\,u(lB)$, since $\widehat{w}(z)=0$ when $z\leq b$. We observe that  $F(b)= hb-mu(lB)=\frac{mu(lB)r\alpha_1}{(\rho+m)(\alpha_1-1)}-mu(lB) <0$ if $\rho+m>r$. Therefore, $F'(z)=h>0$ and $F(z)\leq 0$ for all $z\leq b$, due to $F(b)<0$. To prove (\ref{4-16}), we notice that on $(b,\infty), \widehat{w}''(z)=\frac{-h(\alpha_1-1)}{r}\Big( \frac{m\, u(lB)r \alpha_1}{(\rho+m)h(\alpha_1-1)}\Big)^{1-\alpha_1}z^{\alpha_1-2}>0$ and $\widehat{w}'(b)=0$, so that we have $\widehat{w}'(z)\geq 0, z>b$. Since $\widehat{w}(b)=0$, it follows that $\widehat{w}(z)\geq 0$ for all $z>b$.
	
\vspace{6pt}
\textbf{Step 2:} We verify the optimality of $\widehat{w}(z)$ and of the stopping time $\eta(z;b)$ in (\ref{etab}). Note that $\widehat{w}$ in (\ref{vz})  is $C^2$ on $(0,b) \cup (b,\infty)$, but only $C^1$ at $b$. Let $z>0 $ be given and fixed. We first show that $\widehat{w}(z)\geq \widehat{v}(z)$. \begin{sloppypar}  Applying It\^o's formula to the process $\{e^{-(\rho+m)t}\widehat{w}(Z_t), t \geq 0\}$, we find that 
	\begin{align}\label{4-17}
	e^{-(\rho+m)t}\widehat{w}(Z_{t}) =&\  \widehat{w}(z)+ \int^{t}_0 e^{-(\rho+m)s}\Big(\mathbb{L}_Z -(\rho+m)\Big)\widehat{w}(Z_s)\mathds{1}_{\{Z_s \neq b\}} ds\nonumber \\
	 &- \int^{t}_0 e^{-(\rho+m)s} \theta Z_s \widehat{w}'(Z_s)dW_s.
	\end{align}
	\end{sloppypar}
	The HJB equation (\ref{hjb1}) guarantees that $[\mathbb{L}_Z-(\rho+m)]\widehat{w}(z)\leq -hz+m\,u(lB) $ everywhere on $(0,\infty)$ but $b$. Since $\mathbb{P}_z(Z_t=b)=0$ for all $t$ and all $z$, we then obtain from (\ref{4-17}) that
	\begin{align}\label{4-18}
	e^{-(\rho+m)t}\widehat{w}(Z_{t}) \leq &\  \widehat{w}(z)+  \int^{t}_0 e^{-(\rho+m)s}\Big(-hZ_s+m\,u(lB)\Big) ds\nonumber\\
	 &- \int^{t}_0 e^{-(\rho+m)s} \theta Z_s \widehat{w}'(Z_s)dW_s.
	\end{align}	
	\begin{sloppypar} Let $(\nu_n)_{n\geq 1}$ be a localization sequence of (bounded) stopping times $(\nu_n)_{n}$ diverging to infinity as $n \uparrow \infty$ for the continuous local martingale $\{ \int^t_0 e^{-(\rho+m)s} \theta Z_s \widehat{w}'(Z_s)dW_s, t \geq 0\}$. Then for every stopping time $\eta$ of $Z$ we have by (\ref{4-18}) above 
		\begin{align*}
	e^{-(\rho+m)(\nu_n \wedge \eta)}\widehat{w}(Z_{\nu_n \wedge \eta}) \leq &\  \widehat{w}(z)+ \int^{\nu_n \wedge \eta}_0 e^{-(\rho+m)s}\Big(-hZ_s+m\,u(lB)\Big) ds\nonumber\\
	 &-\int^{\nu_n \wedge \eta}_0 e^{-(\rho+m)s} \theta Z_s \widehat{w}'(Z_s)dW_s,
	\end{align*}	
		for all $n \geq 1$. Taking the $\mathbb{P}_z$-expectation, using the optional sampling theorem to conclude that $\mathbb{E}_z[\int^{\nu_n \wedge \eta}_0 e^{-(\rho+m)s} \theta Z_s \widehat{w}'(Z_s)dW_s]=0$ for all $n$, and letting $n \to \infty$, we find by Fatou's lemma that
 \begin{align}\label{4-19}
		\widehat{w}(z) \geq \mathbb{E}_z[e^{-(\rho+m)\eta}\widehat{w}(Z_{\eta})]+ \mathbb{E}_z \bigg[\int^{\eta}_0 e^{-(\rho+m)s}(hZ_s-m\,u(lB)) ds\bigg]	.
		\end{align}
	Thus, by arbitrariness of $\eta \in \mathcal{S}$ and from (\ref{hatjz}), we find $\widehat{w}(z) \geq \widehat{v}(z)$ for all $z>0$.
	\end{sloppypar}

	Now consider the stopping time $\eta(z;b)$ defined in (\ref{etab}).  We observe that the inequality in  (\ref{4-18}), therefore also in (\ref{4-19}), becomes an equality. Moreover, $\widehat{w}(Z_{\eta(z;b)})=0$. Hence 
 \begin{align*}
\widehat{w}(z) =\mathbb{E}_z[e^{-(\rho+m)\eta(z;b)}\widehat{w}(Z_{\eta(z;b)})]+ \mathbb{E}_z \bigg[\int^{\eta(z;b)}_0 e^{-(\rho+m)s}\Big(hZ_s-m\,u(lB)\Big) ds\bigg]\leq \widehat{v}(z).	
\end{align*}
This shows that $\widehat{w}(z)=\widehat{v}(z)$ for all $z>0$ and $\eta(z;b)$ is optimal. 

\end{proof}

{
\subsection{Proof of Theorem \ref{dualrelation2}}\label{proofdualrelation2}
\begin{proof}

First, we show the duality relations. Since $(c,\pi, \tau) \in \mathcal{A}(x,y)$ is arbitrary, taking the supremum over $(c,\pi,\tau)  \in \mathcal{A}(x,y)$ on the left-hand side of (\ref{dual2}), we get, for any $z>0, x>-\frac{y}{\kappa}$, 
\begin{align*}
{V}(x,y) &\leq v(z,y)+ z x, 
\end{align*}
so that $V(x,y) \leq \inf_{z>0}[ v(z,y)+zx]$ and $ v(z,y) \geq \sup_{x>-\frac{y}{\kappa}} [V(x,y)-zx]$.

Now we consider the inverse inequality. It follows from Proposition \ref{convex} that for given $x>-\frac{y}{\kappa}$ there exists  a unique $z^*(x,y)>0$ such that $x=-v_z(z^*(x,y),y)$. 

Let $y\in \mathbb{R}_+$ be given. For given $z>0$ and $t<\eta^*(z)$ (cf.\ (\ref{etab})), let us define $\widehat{\mathcal{X}}(z,y)$ by 
\begin{align}\label{eqa13-1}
	\widehat{\mathcal{X}}(z,y):= \mathbb{E}_{z,y}\bigg[  \int_0^{\eta^*(z)}\xi_{s}\Big(\mathcal{I}^u(zP_s)-Y_s\Big) ds  + \xi_{\eta^*(z)}\Big(-\widehat{Q}_z(Z_{\eta^*(z)},Y_{\eta^*(z)})\Big)                      \bigg].              	
		\end{align}

From Theorem  \ref{dualrelation1}, we know that $X^*_{\eta^*(z^*)}=-\widehat{Q}_z(Z^*_{\eta^*(z^*)},Y_{\eta^*(z^*)})+\frac{h}{r}$. For any $(x,y)\in \mathcal{O} $, we introduce the candidate optimal consumption process $c_s^*=\mathcal{I}^u(z^*(x,y)P_s)$.

It follows that
\begin{align*}
	z^*\widehat{\mathcal{X}}(z^*,y)=&\ z^*  \mathbb{E}_{z^*,y}\bigg[  \int_0^{\eta^*(z^*)}\xi_{s}\Big(c_s^*-Y_s\Big) ds  + \xi_{\eta^*(z^*)}X_{\eta^*(z^*)}                       \bigg]\\
	=&\   \mathbb{E}_{z^*,y}\bigg[  \int_0^{\eta^*(z^*)}e^{-(\rho+m)s}Z^*_s\Big(c_s^*-Y_s\Big) ds  + e^{-(\rho+m)\eta^*(z^*)} Z^*_{\eta^*(z^*)}X_{\eta^*(z^*)}                       \bigg] \\
	=&\         	  \mathbb{E}_{z^*,y}\bigg[  \int_0^{\eta^*(z^*)}e^{-(\rho+m)s}\Big(u(c_s^*)-\widehat{u}(Z^*_s)-Z^*_sY_s\Big) ds  \\&+ e^{-(\rho+m)\eta^*(z^*)}   \widehat{V}(X_{\eta^*(z^*)},Y_{\eta^*(z^*)})-\widehat{Q}(Z_{\eta^*(z^*)},Y_{\eta^*(z^*)})+Z_{\eta^*(z^*)} \frac{h}{r}                   \bigg] \\	
	=&\         	  \mathbb{E}_{z^*,y}\bigg[  \int_0^{\eta^*(z^*)}e^{-(\rho+m)s}u(c_s^*) ds  + e^{-(\rho+m)\eta^*(z^*)}   \widehat{V}(X_{\eta^*(z^*)},Y_{\eta^*(z^*)})                  \bigg] -  v(z^*,y).	\end{align*}
From this, we deduce that 
\begin{align*}
	v(z^*,y)+z^*\widehat{\mathcal{X}}(z^*,y)&= \mathbb{E}_{z^*,y}\bigg[  \int_0^{\eta^*(z^*)}e^{-(\rho+m)s}u(c_s^*) ds  + e^{-(\rho+m)\eta^*(z^*)}   \widehat{V}(X_{\eta^*(z^*)},Y_{\eta^*(z^*)})                  \bigg]	\\
	& \leq \sup_{(c,\pi,\eta)\in \mathcal{A}(\mathcal{\widehat{\mathcal{X}}}(z^*,y),y)} \mathbb{E}_{\widehat{\mathcal{X}}(z^*,y),y}\bigg[  \int_0^{\eta}e^{-(\rho+m)s}u(c_s) ds  + e^{-(\rho+m)\eta}   \widehat{V}(X_{\eta},Y_{\eta})                  \bigg]	\\
	&\leq \inf_{z>0}(v(z,y)+z\widehat{\mathcal{X}}(z^*,y)) \leq v(z^*,y)+z\widehat{\mathcal{X}}(z^*,y).
		\end{align*}
		Thus, we deduce that 
		\begin{align*}
			V(\widehat{\mathcal{X}}(z^*,y),y)=\inf_{z>0}(v(z,y)+z\widehat{\mathcal{X}}(z^*,y)) =v(z^*,y)+z^* \widehat{\mathcal{X}}(z^*,y).		\end{align*}
			
			Furthermore, since the unique of $z^*$, which depends on $(x,y)$, is strictly decreasing in $x$, we have 
			\begin{align}\label{eqa-14-1}
				\widehat{\mathcal{X}}(z^*,y)=x.
							\end{align}
Thus, we have 
\begin{align*}
	V(x,y)=      	  \mathbb{E}_{z^*,y}\bigg[  \int_0^{\eta^*(z^*)}e^{-(\rho+m)s}u(c_s^*) ds  + e^{-(\rho+m)\eta^*(z^*)}   \widehat{V}(X_{\eta^*(z^*)},Y_{\eta^*(z^*)})                  \bigg].
	\end{align*}
Moreover, the  candidate optimal strategy $c_s^*=\mathcal{I}^u(z^*(x,y)P_s)$
is optimal. Lemma \ref{budget} guarantees the existence of a candidate optimal portfolio process $\pi^*_s$ such that $(c^*,\pi^*) \in \mathcal{A}(x,y)$. 
From (\ref{eqa13-1}), (\ref{eqa-14-1}) and the strong Markov property we have 
\begin{align*}
	X_t^*= \widehat{\mathcal{X}}(Z_t^*,Y_t)=-v_z(Z_t^*,Y_t).
\end{align*}
Since the dual value function $v$ is smooth in the waiting  region $\mathcal{C}$, we can apply It\^o's formula to the optimal wealth process $X_t^*$ for $0\leq t< \eta^*(z^*) $,
\begin{align*}
	dX_t^*=[-v_{zz}(\rho-r+m)Z_t^*-v_{zy}\mu_y Y_t -\frac{1}{2}v_{zzz}{Z_t^*}^2\theta^2-\frac{1}{2}v_{zyy}\sigma_y^2Y_t^2 \\+v_{zzy}\theta Z_t Y_t \sigma_y]dt+ (\theta v_{zz} Z_t^*-\sigma_y v_y Y_t)dW_t.
\end{align*} 
Comparing the above equation with the wealth process (\ref{wealth}), we have 
\begin{align*}
	\pi^*_t= \frac{\theta v_{zz} Z_t^*-\sigma_y v_y Y_t}{\sigma}.
\end{align*}

\end{proof}
}

\subsection{Proof of Theorem \ref{optimal}}\label{proofoptimal}
\begin{proof}
The proof is organized in two steps.

\textbf{Step 1:} We start by giving explicit expressions of the value function in terms of the primal variables. Firstly, we compute $z^*(x,y)=\mathcal{I}^v(-x,y)$, where $\mathcal{I}^v(\cdot,y)$ is the inverse function of $v_z(\cdot,y)$. From (\ref{W2}), (\ref{4-3}) and (\ref{vz}), we obtain 
\begin{equation}\label{a1}
v(z,y)= \left\{
\begin{aligned}
C_1  z^{\alpha_1}+ z^{\frac{\gamma-1}{\gamma}}\frac{\gamma}{(1-\gamma)K} + \frac{zy}{\kappa}, \quad \text{if} \ & z>b,\\
z^{\frac{\gamma-1}{\gamma}}\frac{\gamma}{(1-\gamma)K}+\frac{yz}{\kappa}-\frac{hz}{r}+\frac{m\,u(lB)}{\rho+m},  \quad \text{if} \ & 0<z\leq b,
\end{aligned}
\right.
\end{equation}
  and
\begin{equation}\label{a2}
v_{z}(z,y)= \left\{
\begin{aligned}
C_1 \alpha_1 z^{\alpha_1-1}-\frac{z^{-\frac{1}{\gamma}}}{K}+\frac{y}{\kappa}, \quad \text{if} \ & z>b,\\
 -\frac{z^{-\frac{1}{\gamma}}}{K}+\frac{y}{\kappa}-\frac{h}{r},  \quad \text{if} \ & 0<z\leq b.
\end{aligned}
\right.
\end{equation}
Then from (\ref{a2}) we know
\begin{equation}\label{a3}
z^{*}(x,y)= 
 \Big[(x-\frac{h}{r}+\frac{y}{\kappa})K\Big]^{-\gamma},  \quad \text{if} \ x\geq {\widehat{b}(y)},
\end{equation}
and, if $-\frac{y}{\kappa}<x<{\widehat{b}(y)}$, $z^*{(x,y)}$ satisfies
\begin{align}\label{a4}
 C_1 \alpha_1 (z^*)^{\alpha_1-1}{(x,y)}-(z^*)^{\frac{-1}{\gamma}}{(x,y)}\frac{1}{K}+\frac{y}{\kappa}+x=0.
\end{align}

From (\ref{5-1}), (\ref{a1}) and (\ref{a3}) we thus find
\begin{equation*}
V(x,y)= \left\{
\begin{aligned}
C_1  (z^*)^{\alpha_1}(x,y)+\frac{\gamma (z^*)^{\frac{\gamma-1}{\gamma}}(x,y)}{(1-\gamma)K} + (\frac{y}{\kappa}+x)z^*(x,y), \quad \text{if} \ &-\frac{y}{\kappa}< x< {\widehat{b}(y)},\\
\frac{(x-\frac{h}{r}+\frac{y}{\kappa})^{1-\gamma}K^{-\gamma}}{1-\gamma} + \frac{m\,u(lB)}{\rho+m},\quad \text{if} \ & x \geq {\widehat{b}(y)},
\end{aligned}
\right.
\end{equation*}
where $z^*$ satisfies (\ref{a4}).

\vspace{6pt}




\textbf{Step 2:} Next we give the explicit solutions for the optimal policies. We set ${\widehat{b}(y)}:= x^*(b,y),$ where $x^*( \cdot,y)$ is the inverse function of $z^*(\cdot,y).$ Since $v_z( z^*(x,y),y) =-x$, by taking $x =x^*(z,y)$, computations show that
\begin{align*}
v_z(z,y) =  v_z(z^*(x^*(z,y),y),y) = -x^*(z,y).
\end{align*}
Hence, from (\ref{5-3}) and (\ref{a2}) we have 
\begin{align*}
{\widehat{b}(y)} &= x^*(b,y) = -  v_z(b,y) = \frac{b^{-\frac{1}{\gamma}}}{K}-\frac{y}{\kappa}+\frac{h}{r},
\end{align*}
where $b$ is given by (\ref{b}).

To give the expression of optimal portfolio $\pi^*$, from Theorems \ref{dualrelation1} and \ref{dualrelation2} we have 
\begin{equation*}
\pi^*(x,y)= \left\{
\begin{aligned}
\frac{\theta z^*(x,y)v_{zz}(z^*(x,y),y)-\sigma_y y v_{zy}(z^*(x,y),y)}{\sigma}, \quad \text{if} \ & -\frac{y}{\kappa}<x< {\widehat{b}(y)},\\
\frac{\theta (z^*(x,y))^{-\frac{1}{\gamma}}}{K \gamma \sigma}-\frac{\sigma_y y}{\sigma \kappa}, \quad \text{if} \ & x \geq {\widehat{b}(y)}.
\end{aligned}
\right.
\end{equation*}

By (\ref{a2}), direct calculations show that 
\begin{equation*}
v_{zz}(z,y)= \left\{
\begin{aligned}
C_1 \alpha_1 (\alpha_1-1)z^{\alpha_1-2}+ \frac{z^{-\frac{1}{\gamma}-1}}{K \gamma}, \quad \text{if} \ & z>b,\\
 \frac{z^{-\frac{1}{\gamma}-1}}{K \gamma}, \quad \text{if} \ & 0<z\leq b,
\end{aligned}
\right.
\end{equation*}
and
\begin{equation*}
v_{zy}(z,y)=\frac{1}{\kappa}, \quad \text{for all} \ z>0.
\end{equation*}
Therefore, combining the above expressions we get
\begin{equation*}
\pi^*(x,y)= \left\{
\begin{aligned}
\frac{\theta[C_1 \alpha_1(\alpha_1-1)(z^*)^{\alpha_1-1}(x,y)+\frac{(z^*)^{-\frac{1}{\gamma}}(x,y)}{K \gamma}]- \frac{\sigma_y y}{\kappa}}{\sigma}, \quad \text{if} \ & -\frac{y}{\kappa}<x< {\widehat{b}(y)},\\
\frac{\theta  (x-\frac{h}{r}+\frac{y}{\kappa}) \frac{1}{ \gamma}-  \frac{\sigma_y y}{\kappa}}{\sigma}, \quad \text{if} \ & x \geq {\widehat{b}(y)},
\end{aligned}
\right.
\end{equation*}
where $z^*(x,y)$ is given in (\ref{a4}).

Then, we give the expression of optimal consumption $c^*$. Since $\mathcal{I}^u(x)=x^{-\frac{1}{\gamma}}$, we find 
\begin{equation*}
c^*(x,y)= \left\{
\begin{aligned}
(z^*)^{-\frac{1}{\gamma}}(x,y), \quad \text{if} \ & -\frac{y}{\kappa}<x< {\widehat{b}(y)},\\
K(x+\frac{y}{\kappa}-\frac{h}{r}), \quad \text{if} \ & x \geq {\widehat{b}(y)},
\end{aligned}
\right.
\end{equation*}
where $z^*(x,y)$ is given in (\ref{a4}).

Let now $Z^*$ be a solution of SDE (\ref{Z}) with initial value $Z_0=z^*$. 
From Theorems \ref{dualrelation1} and \ref{dualrelation2}, for $X_t \geq \widehat{b}(Y_t)$, we have 
\begin{align}\label{a-19}
	X^*_t=\frac{(Z^*_t)^{-\frac{1}{\gamma}}}{K}+\frac{h}{r}-\frac{Y_t}{\kappa}.
\end{align}
Similarly, for $-\frac{Y_t}{\kappa}<X_t<\widehat{b}(Y_t)$, we have 
\begin{align}\label{a-20}
	X^*_t= -C_1 \alpha_1 (Z^*_t)^{\alpha_1-1}+\frac{(Z^*_t)^{-\frac{1}{\gamma}}}{K}-\frac{Y_t}{\kappa}.
	\end{align}

\end{proof}

\subsection{Proof of Proposition \ref{pro4.1} }\label{proofpro4.1}
\begin{proof}
	 By (\ref{uhat}) and (\ref{6-8}) we rewrite $\widehat{Q}^B(z,y)$ as follows
\begin{align}\label{6-12}
\widehat{Q}^B(z,y)&=  \mathbb{E}_{z,y}\bigg[    \int_0^{\infty }  e^{-(\rho+m)s} \frac{\gamma}{1-\gamma}{Z_s}^{\frac{\gamma-1}{\gamma}}  ds  \bigg]+\mathbb{E}_{z,y}\bigg[ \int^\infty_0 e^{-(\rho+m)s} Z_sY_sds  \bigg]  \nonumber\\
&=  z^{\frac{\gamma-1}{\gamma}} \frac{\gamma}{1-\gamma}\int_0^{\infty }  e^{-(\rho+m)s}e^{[\frac{\gamma-1}{\gamma} (\rho+m-r-\frac{1}{2}\theta^2)+ \frac{1}{2}\frac{(\gamma-1)^2 \theta^2}{\gamma^2}] s} ds+\frac{zy}{\kappa}, 
\end{align}
where we have used the explicit expression of $Z$. 
Moreover, due to Assumption \ref{assume2} and (\ref{6-12}), we can verify that
\begin{align}\label{6-13}
\widehat{Q}^B(z,y)=  z^{\frac{\gamma-1}{\gamma}} \frac{\gamma}{1-\gamma}\frac{1}{K}+ \frac{zy}{\kappa} < \infty,
\end{align} 
where $K= \frac{1}{\gamma}( \rho+m-r(1-\gamma)-\frac{1-\gamma}{2\gamma} \theta^2)>0$ due to Assumption \ref{assume2}. Finally, it is easy to see that $\widehat{Q}^B \in  C^{2,2}(\mathbb{R}^2_+)$ from (\ref{6-13}) and it satisfies (\ref{6-9}) by direct calculations.
\end{proof}

\subsection{Proof of Theorem \ref{optimal1}}\label{proofoptimal1}
\begin{proof}
The proof is organized in three steps.

\textbf{Step 1:} From Theorem \ref{main}, we know that $\eta^*=0$. It means that $V^B(x,y)=\widehat{V}^B(x,y)$.  For any $(x,y)\in \mathcal{O}$, we have that $\widehat{V}^B(x,y)=\inf_{z>0}[\widehat{Q}^B(z,y)+zx+\bar{u}(z)]$ by Theorem \ref{thedual}. Moreover, it is easy to check that $\widehat{Q}^B(z,y)+\bar{u}(z)$ is strictly convex with respect to $z$ (cf.\ (\ref{6-13})). Then there exists a unique solution $z^*(x,y)>0$ such that 
\begin{align}\label{a19}
	\widehat{V}^B(x,y)=\widehat{Q}^B(z^*(x,y),y)+z^*(x,y)x+\bar{u}(z^*(x,y)),
	\end{align}
where $z^*(x,y):= \mathcal{I}^Q(-x,y)$ and $\mathcal{I}^Q(\cdot,y)$ is the inverse function of $(\widehat{Q}^B_z+\bar{u}_z)(\cdot,y)$. Moreover, $z^* \in C(\mathcal{O}),$ and $z^*(x,y)$ is strictly decreasing with respect to $x$, which is a bijection form. Hence, for any $y\in \mathbb{R}_+$, $z^*(\cdot,y)$ has an inverse function $x^*(\cdot,y)$, which is continuous, strictly decreasing, and maps $\mathbb{R}_+$ to $(-\frac{y}{\kappa},\infty)$.

\vspace{6pt}
\textbf{Step 2:} We now give explicit expressions of the value function in terms of the primal variables. Firstly, we compute $z^*(x,y)=\mathcal{I}^Q(-x,y)$. From (\ref{6-7}) and (\ref{6-13}), we obtain 
\begin{align*}
	\widehat{Q}^B_z(z,y)=-z^{-\frac{1}{\gamma}}\frac{1}{K}+\frac{y}{\kappa}, \quad \bar{u}_z(z)=-m l^{\frac{1-\gamma}{\gamma}} r^{\frac{1-\gamma}{\gamma}}(\rho+m)^{-\frac{1}{\gamma}}z^{-\frac{1}{\gamma}},
	\end{align*}
and 
\begin{align}\label{a20}
	z^*(x,y)=\bigg[\frac{\frac{y}{\kappa}+x}{\frac{1}{K}+m(lr)^{\frac{1-\gamma}{\gamma}}(\rho+m)^{\frac{-1}{\gamma}} }\bigg]^{-\gamma}.
\end{align}
From (\ref{6-13}), (\ref{a19}) and (\ref{a20}), we thus find
\begin{align*}
	&{V}^B(x,y)=\widehat{V}^B(x,y)=\frac{1}{ \frac{1}{K}+m(lr)^{\frac{1-\gamma}{\gamma}}(\rho+m)^{-\frac{1}{\gamma}}} \frac{(\frac{y}{\kappa}+x)^{1-\gamma} }{1-\gamma}.
	\end{align*}

\vspace{6pt}
\textbf{Step 3:} Here we show the optimal policies. Let $Z^*$ be a solution of SDE (\ref{Z}) with initial value $Z_0=z^*$, where $z^*$ is given in (\ref{a20}). From Theorem \ref{thedual}, we already know that the existence of a candidate optimal portfolio process $\pi^*$ such that $(c^*,\pi^*,B^*) \in \mathcal{A}_0^B(x,y)$ and 
\begin{align}\label{4-15-1}
 \mathbb{E}_{x,y}\bigg[  \int_0^{\infty} \xi_{s}\Big(c^*_s-Y_s+mB^*_0\Big) ds                         \bigg]  = x,
\end{align}
 holds, where $c_s^*=\mathcal{I}^u(Z_s)$ is the candidate optimal consumption process, $B^*_0=[\frac{z(\rho+m)}{r}]^{-\frac{1}{\gamma}}l^{\frac{1-\gamma}{\gamma}}$ is the candidate optimal bequest. Moreover, by Theorem 3.6.3 in \cite{karatzas1998methods} or Lemma 6.2 in \cite{karatzas2000utility}, one can then show that $(c^*, \pi^*,B^*)$ is optimal for the problem $\widehat{V}^B$. It thus remains only to find the expressions of optimal portfolio $\pi^*$.

In fact, from (\ref{4-15-1}) we have 
\begin{align}\label{a23}
	X_t^*= \frac{1}{\xi_t} \mathbb{E}_{x,y}\bigg[\int^\infty_t \xi_s(c_s+mB_0-Y_s)\bigg| \mathcal{F}_t\bigg]=\frac{(Z^*_t)^{-\frac{1}{\gamma}}}{K}+\frac{mB_0^*}{r}-\frac{Y_t}{\kappa}.
\end{align}
Applying It\^o's formula to (\ref{a23}), we have 
\begin{align*}
	dX_t^*=\bigg[\frac{Z_t^{-\frac{1}{\gamma}}}{K}\Big(-\frac{1}{\gamma}(\rho-r+m)+(\frac{1+\gamma}{2\gamma^2})\theta^2\Big) -\frac{\mu_y Y_t}{\kappa}    \bigg]dt+\bigg[Z_t^{-\frac{1}{\gamma}}\frac{\theta}{K \gamma}-\frac{\sigma_y Y_t}{\kappa} \bigg]dW_t.
\end{align*}
Then comparing $X^*$ above with (\ref{stateeq}), we can find that
\begin{align}\label{portfolio}
	\pi^*_t=\frac{\theta(Z_t^*)^{-\frac{1}{\gamma}}}{\gamma K \sigma}-\frac{\sigma_y Y_t}{\kappa \sigma},
\end{align}
and the optimal wealth is indeed induced by the strategies $(c^*,\pi^*,B^*).$

Finally, we give the explicit solutions for the optimal polices. From (\ref{a20}) and (\ref{portfolio}) we deduce that 
\begin{align*}
	\pi^*(x,y)=\frac{\theta(z^*)^{-\frac{1}{\gamma}}}{\gamma K \sigma}-\frac{\sigma_y y}{\kappa \sigma}=\bigg[\frac{\frac{y}{\kappa}+x}{\frac{1}{K}+m(lr)^{\frac{1-\gamma}{\gamma}}(\rho+m)^{\frac{-1}{\gamma}} }\bigg]\frac{\theta}{\gamma K \sigma}-\frac{\sigma_y y}{\kappa \sigma}.
		\end{align*}
		Since $\mathcal{I}^u(x)=x^{-\frac{1}{\gamma}}$, then we have  
\begin{equation*}
c^*(x,y)=\mathcal{I}^u(z^*)=\bigg[\frac{\frac{y}{\kappa}+x}{\frac{1}{K}+m(lr)^{\frac{1-\gamma}{\gamma}}(\rho+m)^{\frac{-1}{\gamma}} }\bigg].
\end{equation*}
Moreover, the optimal bequest amount is given by 
\begin{align*}
	B^*_0=\Big[\frac{z^*(\rho+m)}{r}\Big]^{-\frac{1}{\gamma}}l^{\frac{1-\gamma}{\gamma}}=\bigg[\frac{\frac{y}{\kappa}+x}{\frac{1}{K}+m(lr)^{\frac{1-\gamma}{\gamma}}(\rho+m)^{\frac{-1}{\gamma}} }\bigg]\Big[\frac{\rho+m}{r}\Big]^{-\frac{1}{\gamma}}l^{\frac{1-\gamma}{\gamma}}.
	\end{align*}

\end{proof}

{
\subsection{Proof of Theorem \ref{theorem5-5}}\label{prooftheorem5-5}

\begin{proof}
	
	First we show that $\widetilde{w}(z)$ in (\ref{dissolution}) satisfies the HJB equation (\ref{dishjb1}). By construction, we only need to show that 
\begin{align}\label{eq-a26}
	\Big(\mathbb{L}_Z-(\rho+m)\Big)\widetilde{w}(z)+mu(q)-K \widetilde{u}(z)+\frac{z K mq}{r}\leq 0, \quad \forall z\leq \widetilde{ b},
	\end{align}
and
\begin{align}\label{eq-a27}
	\widetilde{w}(z)\geq 0, \quad \forall z> \widetilde{b} .
	\end{align}
	To prove (\ref{eq-a26}), we define $F(z):=(\mathbb{L}_Z-(\rho+m))\widetilde{w}(z)+mu(q)-K \widetilde{u}(z)+\frac{z K mq}{r}=mu(q)-K\widetilde{u}(z)+\frac{z K mq}{r}$,  since $\widetilde{w}(z)=0$ when $z\leq \widetilde{b}$. We observe that  $F(b)= mu(q)-K\widetilde{u}(b)+\frac{b K mq}{r}$ and $F'(z)>0$. Therefore, if $F(b)=mu(q)-K\widetilde{u}(b)+\frac{b K mq}{r}\leq 0$, then $F(z)\leq 0$ for any $z\leq \widetilde{b}$. (\ref{eq-a27}) is satisfied by our assumption. 
	
	Then similar to Theorem \ref{verification},   we can verify the optimality of $\widetilde{w}(z)$ and of the stopping time $\eta(z;\widetilde{b})$.	

\end{proof}
}

\subsection{Some auxiliary results}


\begin{lemma}\label{budget2}
Let $x> \frac{h}{r}-\frac{y}{\kappa}$ be given, let $c\geq 0$ be a consumption process satisfying
\begin{align*}
\mathbb{E}_{x,y}\bigg[ \int_0^{\infty} \xi_{s} c_s ds                         \bigg]  = x+\frac{y}{\kappa} - \frac{h}{r}.
\end{align*}
Then, there exists a portfolio process $\pi$ such that the pair $(c,\pi) \in \mathcal{A}_0(x,y)$  and 
\begin{align*}
X^{c,\pi, \tau}_s+g_s >  \frac{h}{r}, \ \text{for} \ s\geq \eta.
\end{align*}
\end{lemma}
\begin{proof}
Let us define $L_s:= \int^s_0  \xi_{u} c_u du $ and consider the nonnegative martingale 
\begin{align*}
M_s:= \mathbb{E}[L_\infty| \mathcal{F}_s], \ s\geq 0.
\end{align*}
According to the martingale representation theorem, there is an $\mathbb{F}$-adapted process $\phi$ satisfying $\int^{\infty}_0 ||\phi_u ||^2 du < \infty$ almost surely and 
\begin{align*}
M_s= M_0+ \int^s_0 \phi_u dW_u=x+\frac{y}{\kappa} -\frac{h}{r}+ \int^s_0\phi_u dW_u ,  \   s\geq 0.
\end{align*} 
Define then the  nonnegative process $X$ by 
\begin{align*}
X_s&:= \frac{1}{\xi_{s}} \mathbb{E}\bigg[\int^{\infty}_s  \xi_{u} c_u du\bigg| \mathcal{F}_s\bigg]+\frac{h}{r}-g_s=\frac{1}{\xi_{s}}[M_s-L_s]+\frac{h}{r}-g_s,
\end{align*}
so that $X_0=x, M_0=x-\frac{h}{r}+\frac{y}{\kappa}.$  It\^o's rule implies
\begin{align*}
d(e^{-rs}X_s)= -c_s e^{-rs}ds-he^{-rs}ds+Y_se^{-rs}ds+ e^{-rs}\pi_s \sigma dW_s,
\end{align*}
where $\pi_s:=\frac{1}{\xi_{s}\sigma}[\phi_s+(M_s-L_s)\theta]$. It is easy to check that $\pi$ satisfies $\int^\infty_0 |\pi_s|^2 ds< \infty$ a.s. (see, e.g., Theorem 3.3.5 in \cite{karatzas1998methods}). We thus conclude that $X_s=X^{c,\pi,\tau}_s$ when $s \geq \eta$, by comparison with (\ref{wealth}). Finally, since $X_s+g_s>\frac{h}{r}$ for $s \geq \eta$, the pair $(c,\pi)$ is admissible, and $X^{c,\pi, \tau}_s +g_s>  \frac{h}{r}, \ \text{for} \ s\geq  \eta.$   

\end{proof}

\begin{lemma}\label{budget}
For any $\eta \in \mathcal{S}$,  let $x+\frac{y}{\kappa}> 0$ be given, let $c \geq 0$ be a consumption process. For any $\mathcal{F}_\eta$-measurable random variable $\phi$ with $\mathbb{P}[\phi > -g_\eta  ] =1$ such that
\begin{align*}
\mathbb{E}_{x,y}\bigg[\xi_{\eta}\phi+ \int_0^{\eta} \xi_{s} c_s ds                         \bigg]  = x+\frac{y}{\kappa},
\end{align*}
there exists a portfolio process $\pi$ such that the pair $(c,\pi)$ is admissible and 
\begin{align*}
X^{c,\pi, \eta}_s +g_s>0, \ \text{for} \ s \leq \eta,\ \phi = X^{c,\pi,\eta}_\eta.
\end{align*}
\end{lemma}
\begin{proof}
\begin{sloppypar}
The proof is similar to Lemma 6.3 in \cite{karatzas2000utility}, and we thus omit details.
\end{sloppypar}
\end{proof}

\section*{Acknowledgments}
Funded by the Deutsche Forschungsgemeinschaft (DFG, German Research Foundation) – Project-ID 317210226 – SFB 1283.

\bibliographystyle{apalike}
\bibliography{lifeinsurace}

\begin{thebibliography}{}

\bibitem[Bayraktar and Young, 2013]{bayraktar2013life}
Bayraktar, E. and Young, V.~R. (2013).
\newblock Life insurance purchasing to maximize utility of household
  consumption.
\newblock {\em North American Actuarial Journal}, 17(2):114--135.

\bibitem[Bernheim and Severinov, 2003]{bernheim2003bequests}
Bernheim, B.~D. and Severinov, S. (2003).
\newblock Bequests as signals: An explanation for the equal division puzzle.
\newblock {\em Journal of Political Economy}, 111(4):733--764.

\bibitem[Bosserhoff et~al., 2022]{bosserhoff2022investment}
Bosserhoff, F., Chen, A., S{\o}rensen, N., and Stadje, M. (2022).
\newblock On the investment strategies in occupational pension plans.
\newblock {\em Quantitative Finance}, 22(5):889--905.

\bibitem[Braun et~al., 2016]{braun2016consumer}
Braun, A., Schmeiser, H., and Schreiber, F. (2016).
\newblock On consumer preferences and the willingness to pay for term life
  insurance.
\newblock {\em European Journal of Operational Research}, 253(3):761--776.

\bibitem[Campbell, 1980]{campbell1980demand}
Campbell, R.~A. (1980).
\newblock The demand for life insurance: An application of the economics of
  uncertainty.
\newblock {\em The Journal of Finance}, 35(5):1155--1172.

\bibitem[Chen et~al., 2021]{chen2021retirement}
Chen, A., Hentschel, F., and Steffensen, M. (2021).
\newblock On retirement time decision making.
\newblock {\em Insurance: Mathematics and Economics}, 100:107--129.

\bibitem[Choi and Shim, 2006]{jin2006disutility}
Choi, K.~J. and Shim, G. (2006).
\newblock Disutility, optimal retirement, and portfolio selection.
\newblock {\em Mathematical Finance}, 16(2):443--467.

\bibitem[Cocco et~al., 2005]{cocco2005consumption}
Cocco, J.~F., Gomes, F.~J., and Maenhout, P.~J. (2005).
\newblock Consumption and portfolio choice over the life cycle.
\newblock {\em The Review of Financial Studies}, 18(2):491--533.

\bibitem[Cox and Huang, 1989]{cox1989optimal}
Cox, J.~C. and Huang, C.-f. (1989).
\newblock Optimal consumption and portfolio policies when asset prices follow a
  diffusion process.
\newblock {\em Journal of Economic Theory}, 49(1):33--83.

\bibitem[Dybvig and Liu, 2010]{dybvig2010lifetime}
Dybvig, P.~H. and Liu, H. (2010).
\newblock Lifetime consumption and investment: retirement and constrained
  borrowing.
\newblock {\em Journal of Economic Theory}, 145(3):885--907.

\bibitem[Elinder et~al., 2018]{elinder2018inheritance}
Elinder, M., Erixson, O., and Waldenstr{\"o}m, D. (2018).
\newblock Inheritance and wealth inequality: Evidence from population
  registers.
\newblock {\em Journal of Public Economics}, 165:17--30.

\bibitem[Ferrari and Zhu, 2023]{ferrari2023optimal}
Ferrari, G. and Zhu, S. (2023).
\newblock Optimal retirement choice under age-dependent force of mortality.
\newblock {\em arXiv preprint arXiv:2311.12169}.

\bibitem[Gompertz, 1825]{gompertz1825nature}
Gompertz, B. (1825).
\newblock On the nature of the function expressive of the law of human
  mortality, and on a new mode of determining the value of life contingencies.
\newblock {\em Philos. Trans. Roy. Soc. London}, 115:513--585.

\bibitem[Hamaaki et~al., 2019]{hamaaki2019intra}
Hamaaki, J., Hori, M., and Murata, K. (2019).
\newblock The intra-family division of bequests and bequest motives: Empirical
  evidence from a survey on japanese households.
\newblock {\em Journal of Population Economics}, 32(1):309--346.

\bibitem[Huang and Milevsky, 2008]{huang2008portfolio}
Huang, H. and Milevsky, M.~A. (2008).
\newblock Portfolio choice and mortality-contingent claims: The general {HARA}
  case.
\newblock {\em Journal of Banking \& Finance}, 32(11):2444--2452.

\bibitem[Karatzas and Kou, 1998]{karatzas1998hedging}
Karatzas, I. and Kou, S.~G. (1998).
\newblock Hedging {American} contingent claims with constrained portfolios.
\newblock {\em Finance and Stochastics}, 2(3):215--258.

\bibitem[Karatzas et~al., 1986]{karatzas1986explicit}
Karatzas, I., Lehoczky, J.~P., Sethi, S.~P., and Shreve, S.~E. (1986).
\newblock Explicit solution of a general consumption/investment problem.
\newblock {\em Mathematics of Operations Research}, 11(2):261--294.

\bibitem[Karatzas et~al., 1987]{karatzas1987optimal}
Karatzas, I., Lehoczky, J.~P., and Shreve, S.~E. (1987).
\newblock Optimal portfolio and consumption decisions for a “small
  investor” on a finite horizon.
\newblock {\em SIAM Journal on Control and Optimization}, 25(6):1557--1586.

\bibitem[Karatzas and Shreve, 1998]{karatzas1998methods}
Karatzas, I. and Shreve, S.~E. (1998).
\newblock {\em Methods of mathematical finance}, volume~39.
\newblock Springer.

\bibitem[Karatzas and Wang, 2000]{karatzas2000utility}
Karatzas, I. and Wang, H. (2000).
\newblock Utility maximization with discretionary stopping.
\newblock {\em SIAM Journal on Control and Optimization}, 39(1):306--329.

\bibitem[Kopczuk, 2013]{kopczuk2013taxation}
Kopczuk, W. (2013).
\newblock Taxation of intergenerational transfers and wealth.
\newblock In {\em Handbook of public economics}, volume~5, pages 329--390.
  Elsevier.

\bibitem[Li et~al., 2007]{li2007demand}
Li, D., Moshirian, F., Nguyen, P., and Wee, T. (2007).
\newblock The demand for life insurance in {OECD} countries.
\newblock {\em Journal of Risk and Insurance}, 74(3):637--652.

\bibitem[McGarry, 2013]{mcgarry2013estate}
McGarry, K. (2013).
\newblock The estate tax and inter vivos transfers over time.
\newblock {\em American Economic Review}, 103(3):478--483.

\bibitem[Merton, 1971]{merton1971optimum}
Merton, R.~C. (1971).
\newblock Optimum consumption and portfolio rules in a continuous-time model.
\newblock volume~3, pages 373--413. Elsevier.

\bibitem[Park et~al., 2023]{park2023robust}
Park, K., Wong, H.~Y., and Yan, T. (2023).
\newblock Robust retirement and life insurance with inflation risk and model
  ambiguity.
\newblock {\em Insurance: Mathematics and Economics}, 110:1--30.

\bibitem[Peskir and Shiryaev, 2006]{peskir2006optimal}
Peskir, G. and Shiryaev, A. (2006).
\newblock {\em Optimal stopping and free-boundary problems}.
\newblock Springer.

\bibitem[Pliska and Ye, 2007]{pliska2007optimal}
Pliska, S.~R. and Ye, J. (2007).
\newblock Optimal life insurance purchase and consumption/investment under
  uncertain lifetime.
\newblock {\em Journal of Banking \& Finance}, 31(5):1307--1319.

\bibitem[Poterba, 2001]{poterba2001estate}
Poterba, J. (2001).
\newblock Estate and gift taxes and incentives for inter vivos giving in the
  us.
\newblock {\em Journal of Public Economics}, 79(1):237--264.

\bibitem[Richard, 1975]{richard1975optimal}
Richard, S.~F. (1975).
\newblock Optimal consumption, portfolio and life insurance rules for an
  uncertain lived individual in a continuous time model.
\newblock {\em Journal of Financial Economics}, 2(2):187--203.

\bibitem[Wang et~al., 2021]{wang2021household}
Wang, N., Jin, Z., Siu, T.~K., and Qiu, M. (2021).
\newblock Household consumption-investment-insurance decisions with uncertain
  income and market ambiguity.
\newblock {\em Scandinavian Actuarial Journal}, 2021(10):832--865.

\bibitem[Wei et~al., 2020]{wei2020optimal}
Wei, J., Cheng, X., Jin, Z., and Wang, H. (2020).
\newblock Optimal consumption--investment and life-insurance purchase strategy
  for couples with correlated lifetimes.
\newblock {\em Insurance: Mathematics and Economics}, 91:244--256.

\bibitem[Zeng et~al., 2016]{zeng2016optimal}
Zeng, X., Carson, J.~M., Chen, Q., and Wang, Y. (2016).
\newblock Optimal life insurance with no-borrowing constraints: duality
  approach and example.
\newblock {\em Scandinavian Actuarial Journal}, 2016(9):793--816.

\end{thebibliography}


\end{document}